\documentclass{article}
\usepackage{ijcai25}

\usepackage{times}
\usepackage{soul}
\usepackage{url}
\usepackage[hidelinks]{hyperref}
\usepackage[utf8]{inputenc}
\usepackage[small]{caption}
\usepackage{graphicx}
\usepackage{amsmath}
\usepackage{amsthm}
\usepackage{booktabs}
\usepackage{algorithm}
\usepackage{algorithmic}
\usepackage[switch]{lineno}
\usepackage{subcaption}

\newtheorem{example}{Example}

\newtheorem{lemma}{Lemma}

\usepackage{amsmath,amssymb,amsfonts,amsthm,breqn,comment}
\usepackage{algorithmic}
\usepackage{graphicx}
\usepackage{textcomp}
\usepackage{xcolor}
\usepackage{bbold}
\usepackage{cleveref}

\usepackage{thmtools}
\usepackage{thm-restate}
\usepackage{framed}
 \usepackage{multirow}

\title{Synthesis of Communication Policies for Multi-Agent Systems Robust to Communication Restrictions}

\author{
Saleh Soudijani$^1$
\and
Rayna Dimitrova$^1$\\\
\affiliations
$^1$CISPA Helmholtz Center for Information Security, Germany\\
\emails
\{saleh.soudijani, dimitrova\}@cispa.de
}

\newcommand{\nats}{\ensuremath{\mathbb{N}}}

\newcommand{\prob}{\ensuremath{\mathbb{P}}}

\newcommand{\init}{\ensuremath{\mathit{init}}}
\newcommand{\states}{\ensuremath{\mathcal{S}}}
\newcommand{\acts}{\ensuremath{\mathcal{A}}}
\newcommand{\starget}{\ensuremath{\mathcal{S}_{\mathit{target}}}}
\newcommand{\savoid}{\ensuremath{\mathcal{S}_{\mathit{avoid}}}}
\newcommand{\tpath}{\ensuremath{\tau}}

\newcommand{\policy}{\ensuremath{\pi}}

\newcommand{\lstates}{\ensuremath{\mathcal{L}}}
\newcommand{\ostates}{\ensuremath{\mathcal{O}}}
\newcommand{\wstarget}{\ensuremath{\widehat{\mathcal{S}}_{\mathit{target}}}}
\newcommand{\wsavoid}{\ensuremath{\widehat{\mathcal{S}}_{\mathit{avoid}}}}

\newcommand{\piact}{\ensuremath{\pi_{act}}}
\newcommand{\Piact}{\ensuremath{\Pi_{act}}}
\newcommand{\Piactpos}{\ensuremath{\Pi^{pos}_{act}}}

\newcommand{\Acomm}{\ensuremath{\mathcal{A}_{comm}}}
\newcommand{\picomm}{\ensuremath{\pi_{comm}}}
\newcommand{\Picomm}{\ensuremath{\Pi_{comm}}}
\newcommand{\Picommpos}{\ensuremath{\Pi_{comm}^{pos}}}

\newtheorem{definition}{Definition}

\begin{document}

\maketitle

\begin{abstract}
We study stochastic multi-agent systems in which agents must cooperate to maximize the probability of achieving a common reach-avoid objective. 
In many applications, during the execution of the system, the communication between the agents can be constrained by restrictions on the bandwidth currently available for exchanging local-state information between the agents. 
In this paper, we propose a method for computing joint action and communication policies for the group of agents that aim to satisfy the communication restrictions as much as possible while achieving the optimal reach-avoid probability when communication is unconstrained. Our method synthesizes a pair of action and communication policies robust to restrictions on the number of agents allowed to communicate. To this end, we introduce a novel cost function that measures the amount of information exchanged beyond what the communication policy allows. We evaluate our approach experimentally on a range of benchmarks and demonstrate that it is capable of computing pairs of action and communication policies that satisfy the communication restrictions if such exist.
\end{abstract}

\section{Introduction}
In cooperative multi-agent systems (MAS),  individual agents are required to collaborate to achieve a joint task.   
One way to achieve this collaboration is to provide a centralized joint policy that the agents must adhere to.  
Typically,  such policies need to correlate the actions of different agents. As a result,  their successful execution requires communication between the agents to exchange, for example, local state information,  and to coordinate their actions.
%
However, in many real-world settings,  agents have to operate in environments where communication could be restricted due to physical limitations, such as limited bandwidth or signal interference.
As such restrictions can severely impact the coordination between agents and, hence, their performance, it is imperative that communication restrictions be considered in the design of joint policies for MAS.
%
This requires devising policies that prescribe how the limited resources available for communication should be allocated.

In this paper, we focus on cooperative MAS with joint reach-avoid objectives (which require the agents to reach some target set of joint states while avoiding some unsafe states),  \emph{possibly operating under communication restrictions on the number of agents allowed to communicate. } 
We consider the setting where communicating agents exchange full current state information and jointly select actions.
This requires policies that determine which subset of agents should communicate at the current state of system execution,  which we term \emph{communication policies}.
Clearly,  communication policies cannot depend on the full information about agents' local states,  as this would defeat their purpose.  In this work, we assume that communication policies can use some \emph{public information} about the agents' states, which, in practice, could be very limited or even non-existent. 
A typical example is a system in which agents know the coarse regions in which other agents are located but not their precise locations. 
The rationale is that this public information is significantly less costly to communicate and changes less frequently.

We study the problem of synthesizing communication policies,  together with joint policies that govern the agents' actions,  which we call \emph{action policies}.
The challenge is that these two policies should be synthesized in tandem since the joint action policy should be adapted to the communication policy, requiring as little communication as possible beyond that allowed by the communication policy.
To address this challenge,  we introduce a cost function that, intuitively,  measures the \emph{information exchange between agents required by the action policy that goes beyond what is allowed by the communication policy}.  We show that this cost function can be used in an upper bound on the performance loss when the action policy is executed under restricted communication and following the communication policy. 
Based on this bound,  we propose a method for synthesizing pairs of action and communication policies that minimizes an over-approximation of the cost function and achieves optimal reach-avoid probability under unrestricted communication.
%

\paragraph{Related work.} 
Decision-theoretic models for MAS~\cite{rizk2018decision} such as decentralized MDPs (Dec-MDPs) and decentralized partially observable MDPs (Dec-POMDPs)~\cite{decentralizedcontrolCS} and their respective policy synthesis problems have been extensively studied.  A key characteristic of these models is that agents cannot communicate.
In contrast, in our setting,  agents are allowed to communicate and exchange information about their independent local states and transitions,  but this communication must be minimized relative to a communication policy.\looseness=-1

A number of MAS models exist where communication is allowed but used sparsely to simplify the policy synthesis task. These include~\cite{guestrin2001multiagent}, where a coordination graph representing the dependencies between the agents is given, and  \cite{sparseinteractionsmelo2011decentralized}, where the decentralized model is equipped with information about the states in which the agents need to interact. 
Other methods~\cite{WuZC11} use online planning to use communication dynamically on demand. 
These approaches enable the synthesis of optimal policies that conform to given communication structures or minimize communication. On the other hand, our method synthesizes optimal action policies equipped with communication policies that make them robust to communication restrictions.\looseness=-1

The closest to our work is~\cite{KarabagNT22},  which proposes a technique for constructing joint policies for cooperative MAS that are robust to temporary or permanent loss of communication. In contrast,  the policies we compute must be \emph{robust to communication restrictions},  and thus benefit from \emph{associated communication policies}. 
Thus, while~\cite{KarabagNT22} can use total correlation to synthesize policies minimizing dependency between the agents, we need to develop a cost function whose values depend on the sought communication policy.  Similarly to~\cite{KarabagNT22},  our cost function uses information-theoretic measures, but the challenge is to account for the unknown communication policy.

\section{Preliminaries}\label{sec:prelim}
In this section, we review some definitions and concepts.

For $n \in \nats$,  we define $[n] := \{1,\ldots,n\}$.
We denote the set of discrete probability distributions over a set \(X\) with \(\Delta(X)\).

Markov decision processes (MDPs) provide a framework for modeling and analysis of sequential decision processes.

\begin{definition}
A \emph{Markov decision process (MDP)} is a tuple \(M = (\states,\acts, P, s_\init)\) where 
\(\states\) is a finite set of states, 
\(\acts\) is a finite set of actions, 
\(P: \states \times \acts \rightarrow \Delta(\states)\) is a partial transition probability function, 
and \(s_\init \in \states\) is an initial state.
\end{definition}

For simplicity we sometimes write $P(s,a,s')$ instead of $P(s,a)(s')$,  for $s,s' \in \states$ and $a\in\acts$. 
We denote with $\acts(s) := \{s \in \states \mid \exists s' \in \states.~P(s,a,s') > 0\}$ the set of actions enabled in $s\in\states$. 
We assume that $\acts(s) \neq \emptyset$ for every $s \in \states$.

A \emph{path} in an MDP \(M\) is a finite or infinite sequence
\(\tpath = s_0 a_1 s_1 \ldots  s_{t-1} a_t s_t,\ldots\)
of alternating states and actions such that 
\(P(s_t, a_{t+1},s_{t+1}) > 0\) for all \(t \in \nats\).  

A \emph{policy} for an MDP \(M = (\states,\acts, P, s_\init)\) is a function 
\(\policy: (\states \cdot \acts)^*\cdot \states \to \Delta(\acts)\),
 that maps each finite path ending in a state to a probability distribution over actions and is such that if $\policy(s_0a_1\ldots s_t)(a) > 0$, then $a \in \acts(s_t)$. 
A policy is called \emph{positional}, if its decisions depend solely on the current state. 
Formally,  we can represent a positional policy \(\policy\) as a function \(\policy: \states \to \Delta(\acts)\).
For simplicity,  we write $\policy(\tau,a)$ and $\policy(s,a)$ instead of $\policy(\tau)(a)$ and $\policy(s)(a)$, respectively. 

\begin{definition} 
A \emph{Markov chain} is a triple \(C = (\states, P, s_\init)\) where 
\(\states\) is the set of states, 
\(P: \states \to \Delta(\states)\) is the transition probability function,  and
\(s_\init \in \states\) is the initial state.
\end{definition}

Given an MDP \(M\), a policy \(\policy\) for \(M\) induces an (potentially infinite-state) Markov chain. 
We denote this Markov chain with \(M_{\pi}\), which is defined as
\(M_{\pi} = ((\states \cdot \acts)^*\cdot \states, P_{M_\pi}, s_\init)\), 
where for every 
\(\tpath=s_0 a_1 \ldots s_{t-1} a_t s_t \in (\states\cdot\acts)^*\cdot\states$,
$a \in \acts$ and $s \in \states$ we have that 
$P_{M_\policy}(\tpath,\tpath\cdot a \cdot s) = \policy(\tpath,a)\cdot P(s_t,a,s) \). 
For a positional policy \(\pi\) for an MDP \(M\), 
the induced Markov chain has a finite set of states. 
Formally,  \(M_{\policy} = (\states, P_{M_\policy}, s_\init)\), where for every $s \in \states$ and $s' \in \states$ we have $P_{M_\policy}(s, s') = \sum_{a\in\acts(s)}\policy(s,a)\cdot P(s,a,s') \).

A Markov chain \(C = (\states, P, s_\init)\) can be seen as a sequence of discrete stochastic variables \((S_t, t \in \nats ) \), which generates a stationary process \(S\) where \(\prob(S_t=s)\) is the probability of the chain visiting state \(s \in \states\) at time \(t\). 
%
The \emph{occupancy measure of a state} $s$ is 
\(\nu_{s} := \sum_{t=0}^{\infty} \prob(S_t=s)\).

Given a policy $\policy$ for an MDP \(M = (\states,\acts, P, s_\init)\),  we denote with $\nu_{s,a}$ the \emph{occupancy measure of the state-action pair} \((s,a)\), i.e.,  the expected number of times that action $a$ is taken at state $s$, defined as $\nu_{s,a} := \sum_{t=0}^{\infty} \prob(S_t=s, A_t=a)$.
By definition,  we have that 
$\nu_{s,a} = 
\sum_{t=0}^{\infty} \prob(S_t=s, A_t=a) = 
\sum_{t=0}^{\infty} \prob(S_t=s)\cdot \prob(A_t=a\mid S_t = s) = \policy(s,a)\cdot\nu_s$.

\subsubsection{Entropy of Stochastic Processes}

The entropy is a measure of uncertainty about the outcome of a random variable \cite{shannon1949mathematical}.

\begin{definition}
For a discrete random variable \(X\), its support \(V\) defines a countable sample space from which \(X\) takes a value \(v \in V \) according to a probability mass function (pmf)  \(p(v):=\prob(X=v)\).
The \emph{entropy} of \(X\) is defined as
$H(X) := - \sum_{v \in V} p(v) \log p(v).$
By convention,  \(0\log0=0\).
\end{definition}
The entropy is always non-negative. It vanishes for a deterministic \(X\) (i.e., if \(X\) is completely determined).

Let \((X_1, X_2)\) be a pair of random variables with joint pmf \(p(v_1, v_2)\) and support \(V_1 \times V_2 \). 
The \emph{joint entropy} of  \((X_1,X_2)\) is defined by
$H(X_1,X_2):=-\sum_{v_1 \in V_1} \sum_{v_2 \in V_2} p(v_1,v_2) \log p(v_1,v_2). $

The \emph{conditional entropy} of a random variable \(X_1\) given \(X_2\) is defined as
$H(X_1 | X_2):=-\sum_{v_1 \in V_1} \sum_{v_2 \in V_2} p(v_1,v_2) \log p(v_1 \mid v_2). $

The joint and conditional entropy definitions extend to the collection of \(n\) random variables \cite{cover1999elements}.

The \emph{entropy of a Markov chain} $C = (\states,P,s_\init)$ is defined as the joint entropy over all random variables $S_t$ for $t \in \nats$. That is,  \(H(C) := H(S_0,S_1,S_2,\ldots)= \sum_{t=0}^\infty H(S_t \mid S_{t-1} \ldots S_0 ) \).
The entropy of a Markov chain is in general infinite. 
The entropy of a Markov chain is finite if and only if it is absorbing \cite{BiondiLNW14}. 
In this paper, we restrict our analysis to absorbing Markov chains. 

\cite{BiondiLNW14} showed that the entropy of a Markov chain can be characterized in terms of the occupancy measure of the states and their so-called \emph{local entropy}. The local entropy $L(s)$ of a state $s$ in a Markov chain is the entropy of the probability distribution over the next states defined by $P$, formally,  
$L(s) :=  H(S_{t+1} \mid S_t = s) = -\sum_{s' \in \states}P(s,s') \log P(s,s')$.
Then, as shown in \cite{BiondiLNW14},  the entropy $H(C)$ can be expressed as 
 \(H(C) =\sum_{s\in \states} L(s)\cdot \nu_s\),  where $\nu_s$ is the occupancy measure of   $s$.

\subsubsection{Multi-Agent Markov Decision Processes}

\emph{Multi-agent Markov decision processes (MMDPs)} describe sequential decision-making tasks in which multiple
agents select actions in order to collaboratively maximize a given common reward-based optimization criterion.
A joint policy prescribes actions for all agents.  
During the execution of such a policy,  all agents have access to the joint state of the system.

\begin{definition}\label{def:mmdp}
Formally,  a \emph{Multi-agent Markov decision process (MMDP)} is a tuple \(M=(N,\states,\acts,P,s_\init)\) where:
\begin{itemize}
\item \(N\) is the number of agents;
\item \(\states=\states^1\times \states^2\times \ldots \times \states^N\) is a finite set of global states;
\item \(\acts=\acts^1\times \acts^2\times \ldots \times \acts^N\) is a finite set of joint actions;
\item \(P:\states\times \acts \rightarrow \Delta(\states)\) is the joint transition probability function such that for every  
$s = \langle s^1,\ldots, s^N\rangle \in \states$, 
$a= \langle a^1,\ldots, a^N\rangle \in \acts$  and 
$u = \langle u^1,\ldots, u^N\rangle \in \states$ it holds that
$P(s,a)(u) := \prod_{i=1}^N P^i(s^i,a^i)(u^i)$, 
where $\states^i$, $\acts^i$ and
$P^i : \states^i \times \acts^i \to \Delta(\states^i)$ are the local states,  actions and transition probability function of agent $i \in [N]$;
\item \(s_\init = \langle s_\init^1,\ldots, s_\init^N\rangle \in \states\) is the initial state.
\end{itemize}
\end{definition}

We identify agents with the elements of the set $[N]$. 
Each agent $i \in [N]$ is modeled by an MDP 
$M^i = (\states^i,\acts^i,P^i,s_\init^i)$.
Note that the transition probability functions $P^i$ of the agents are independent and the joint transition probability function in \Cref{def:mmdp} is defined as their product.
We denote with $\overline{i}$ the set $[N] \setminus \{i\}$ of agents other than $i$. \looseness=-1

All agents have access to a centralized view and operate in the environment by executing a positional \emph{ joint policy} \(\pi_{act}:\states \to \Delta(\acts)\). We denote with $\Piactpos(M):=\states \to \Delta(\acts)$ the set of positional joint policies for an MMDP $M$. 

In this paper,  we consider MMDPs with joint \emph{reach-avoid objectives}.  Such an objective is defined as a pair $(\starget,\savoid)$ of sets of states such that \(\starget \cap \savoid = \varnothing \). 
It requires that the agents maximize the probability of reaching  a joint state in $\starget$, while avoiding $\savoid$.

Given a joint policy $\piact$ for an MMDP $M$ with a reach-avoid objective $(\starget,\savoid)$,  we denote with $\prob_{M_\piact}((\neg\savoid) \mathcal{U} \starget)$ the probability of reaching $\starget$ while avoiding $\savoid$ in the Markov Chain $M_\piact$.

Given a reach-avoid objective in an MMDP $M$,   the optimal joint policy synthesis problem requires finding a policy $\policy^*$ such that 
$\prob_{M_{\policy*}}((\neg\savoid) \mathcal{U} \starget) =  \sup_{\policy} \prob_{M_\policy} ((\neg\savoid) \mathcal{U} \starget)$.
The optimal policy that maximizes the reach-avoid probability is denoted as \(\policy^*\), and the optimal value is denoted as \(v^*(M,\starget,\savoid)\).

Since a joint policy in an MMDP has access to the full state,  it is essentially a policy in the product MDP. 
Thus,  for reach-avoid objectives, it suffices to consider positional joint policies.
An optimal joint policy \(\pi_{act}\) for the reach-avoid objective $(\starget,\savoid)$ in the MMDP $M$ can be computed using standard techniques by solving a linear program.

\section{Problem Formulation}\label{sec:problem}
Implementing joint policies in environments with uncertainty requires effective coordination among the agents. Achieving this coordination often depends on establishing robust communication. 
However, agents may face limitations in their communication capabilities, including restrictions on the type and amount of information they can share. Additionally, at certain time steps, the environment might limit the number of agents allowed to communicate.
Let us consider an example.

\begin{example}
Figure~\ref{fig:Env_1_figure} depicts a simple robot navigation problem where three robots must coordinate to reach their respective targets while avoiding collision.
The target locations of robots \(R1\), \(R2\), and \(R3\) are labeled with \(T1\), \(T2\), and \(T3\), respectively.  Each of \(R1\) and \(R2\) has two potential target locations.  Once each of the robots has reached one of their targets,  the team's task is complete.

To maximize the probability of reaching their targets while avoiding collision, the robots must communicate, exchanging information about their current locations and actions.  What communication is necessary depends on the executed joint policy and the resulting execution paths.
For example, if \(R1\) and \(R2\) decide to swap their locations to reach their respective targets at the top of the grid,  communication between these two agents is essential to avoid collisions, while \(R3\) can navigate to its target independently, disregarding the positions of the other robots. On the other hand, if one of \(R1\) or \(R2\) decides to reach its target at the bottom of the grid, communication between this robot and \(R3\) will be needed. Finally, if both \(R1\) and \(R2\) decide to reach their targets at the bottom of the grid, then all three robots must communicate.

Consider the scenario where the communication  is constrained, and at any given time \emph{at most two robots are allowed to communicate}. 
In order to be robust to this restriction,  the joint policy should \emph{minimize the need for communication between all three robots at the same time},  and should be equipped with a \emph{communication policy} that prescribes which pair of robots should communicate at a given time.

\begin{figure}
    \centering
    \includegraphics[scale=0.07]{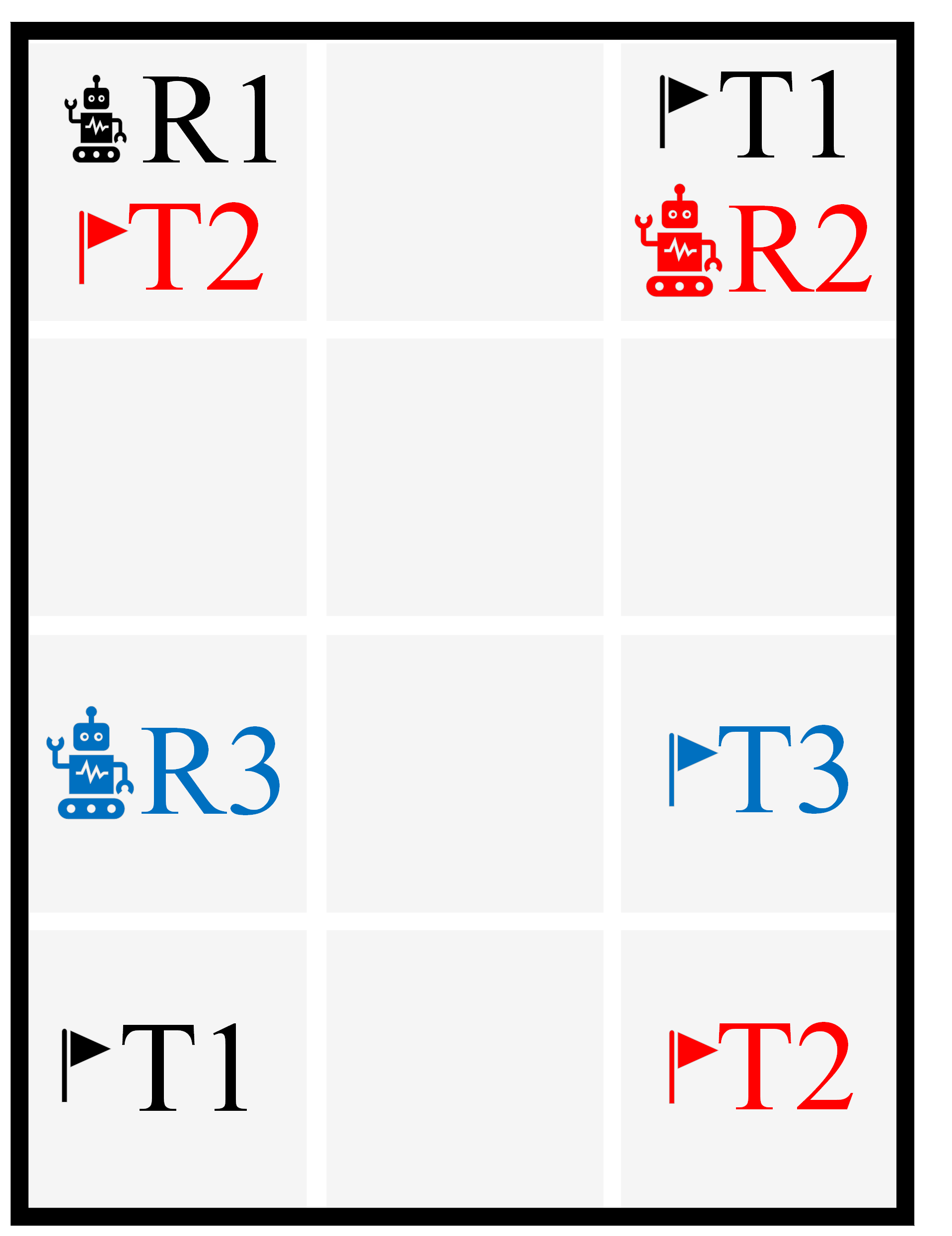}
    \caption{Environment \#1 for a robots navigation problem, with robots \(R1\), \(R2\), and \(R3\) and their respective targets (\(T1\), \(T2\), \(T3\)).}
    \label{fig:Env_1_figure}
\end{figure}

\end{example}

We study the cooperative execution of a joint policy under specific conditions. 
Consider a scenario where agents can, without restriction, share some \emph{public information}, such as, for example,  their current region. 
Additionally, they can share precise state and local information, but there is a limitation on the number of agents permitted to do this at any given time.

We now extend the policy execution in this scenario, which we refer to as \emph{'restricted communication'},  and present our problem formulation. 
We formalize our problem as a Markov game with one and a half players: the multi-agent system, and the stochastic environment. 
The objective of the system is to reach the set of target states while avoiding the ''avoid`` states.
The game is played in a sequence of rounds, starting at an initial state. 
At each step,  public information is freely exchanged between all agents,  and based on this information the  agents collaboratively select a subset of agents for further information exchange,  which includes the sharing of agents' local states. 
Communication is established within this subset, with the selected agents sharing information and making collective decisions. 
Conversely, the remaining agents cannot communicate and share local state information. 
Consequently, the remaining agents must act independently and make decisions solely based on locally available information and estimates of other agents' states.  
After all agents execute the respective policies to select actions,  the process transitions to the next state according to the probabilistic transition relation.

\begin{definition}\label{def:game}
We formulate the team's decision problem as a \emph{cooperative Markov game} represented as a tuple
\(\widehat{M}=(M,\ostates^1,\ldots, \ostates^N, \lstates^1,\ldots \lstates^N,K)\) where 
\begin{itemize}
\item for each $i \in [N]$,  $\ostates^i$ and $\lstates^i$ are finite sets of respectively \emph{publicly observable} and \emph{local states} of agent $i$;
\item $M = (N,\states,\acts,P,s_\init)$ is an MMDP such that \(\states^i=\ostates^i\times \lstates^i\) for each agent $i \in [N]$;
\item $K \in \{0,1,\ldots, N\}$ is the number of agents allowed to communicate at each point in time.
\end{itemize}
\end{definition}
In the above definition, we consider MMDPs of a special form where the states of each agent are factored into a \emph{public} part $\ostates^i$ that can be observed by all the other agents and \emph{local} part $\lstates^i$ that cannot be directly observed by the other agents.  In order for agent $i$'s local state information to become known to another agent,  agent $i$ must communicate that information to that agent.  
We consider a setting where there can be a disruption in the communication, 
resulting in the restriction that only $K$ out of the $N$ agents are allowed to communicate.
When $K=N$ we are in the full-communication case where any joint policy can be executed due to the unrestricted communication between the agents. 
In the other extreme,  when $K=0$ no communication is allowed. 
When communication is restricted,  that is,  $K < N$,  we assume that agents rely on the notion of \emph{imaginary play} introduced in~\cite{KarabagNT22} in order to estimate the current local states of other agents.  
Furthermore,  when $0 < K < N$,  agents need to agree at each step of the execution on a subset of at most $K$ agents that will communicate,  that is,  exchange state information, at this step. 
This is done by choosing the so-called \emph{communication actions} prescribed by a \emph{communication policy}. 
The communication policy is a joint policy that is guaranteed to be implementable  because it relies only on the public part of the agents' states, which can be always shared by all agents.

In \Cref{def:game},  the number $K$ of agents allowed to exchange information when communication is restricted is fixed.  Our results can be extended to the case when $K$ changes dynamically in the course of the execution.

We define the set of \emph{communication actions} consisting of the sets of agents of size exactly $K$, that is 
$\Acomm := \{c \subseteq [N] \mid |c| = K\}$.  
While we could allow communication actions selecting sets with fewer than $K$ agents,  such actions are dominated by those with maximal allowed cardinality.

For $c \in \Acomm$,   the set of remaining agents is $\overline{c}:=[N] \setminus c$.

We denote by $\ostates := \ostates^1 \times \ldots \times\ostates^N$ the set of joint public states.
A \emph{communication policy} is a function of the form $\picomm: \ostates^+ \to \Delta(\Acomm)$.  We let $\Picomm(\ostates, K)$ be the class of all communication policies for given $\ostates$ and $K$.
We denote with $\Picommpos(\ostates,  K) := \ostates \to \Delta (\Acomm)$ the set of positional communication policies for $\widehat M$.

The agents operate in the environment by executing a pair of joint policies \(\pi = (\picomm,\piact)\) where $ \picomm \in \Picomm$ and $ \piact \in \Piact$.  
We refer to $\piact$ as an \emph{action policy} and to $\picomm$ as a \emph{communication policy.} 
Each agent maintains a \emph{local imaginary copy} of the current local states of the other agents.
At each decision step, all agents first share the public parts of their states.  
The agents execute jointly the policy \(\picomm\) to select a subset $c$ of $K$ agents that will communicate with each other at the current step.  
The subset of agents eligible for communication shares the local part of their states with each other, and each agent updates their local imaginary copy based on the received information.  
Thus, the agents in $c$ have accurate knowledge of each other's current local state,  while this can be inaccurate for the rest of the agents. 
Subsequently,  the agents in $c$ jointly execute the action policy \(\piact\) to determine a joint action. 
Each of the other agents executes \(\piact\) independently.
Its state gets updated accordingly,  and the local imaginary copies of the states of the agents in $\overline{i}$ are sampled from the respective distributions.
After that,  the system proceeds to the next decision step. 
This continues until a state in $ \starget \cup \savoid $ is reached.
At each step, the agents selected for communication operate in a centralized manner. 

\subsubsection{Execution under Restricted Communication}

The evolution of the system given a pair of communication and action policies can be formalized as follows. 
Given  \(\widehat{M}=(M,\ostates^1,\ldots, \ostates^N, \lstates^1,\ldots \lstates^N,K)\) with $M=(N,\states,\acts,P,s_\init)$, 
a pair of positional joint policies \(\pi = (\picomm,\piact) \in \Picommpos(\ostates,K)\times\Piactpos(M)\) induces a Markov chain 
$\widehat{M}_{\pi} = (\widehat{\states},\widehat{P},\widehat s_\init)$ defined as follows.
\begin{itemize}
\item The set of states is $\widehat{\states} = \widehat{\states}^1 \times \ldots\times\widehat{\states}^N$, where for each agent $i \in [N]$ we have  $\widehat{\states}^i := \states^i \times \prod_{j \in \overline i} \lstates^j$.
\item Let
$\widehat s_1 = \langle \widehat s^1_1,\ldots, \widehat s_1^N\rangle \in \widehat\states$ and 
$\widehat s_2 = \langle \widehat s_2^1,\ldots, \widehat s_2^N\rangle\in\widehat\states$ 
where 
$\widehat s^i_1 = \langle\langle o^i_1,l^i_1\rangle,\langle l_1^{i,j}\rangle_{j \in \overline i}\rangle$ and 
$\widehat s^i_2 = \langle\langle o^i_2,l^i_2\rangle,\langle l_2^{i,j}\rangle_{j \in \overline i}\rangle$ for all $i \in [N]$.
Let
$\widehat P(\widehat s_1,\widehat s_2) = 
\sum_{c \in\Acomm} \sum_{a\in \acts}\sum_{\substack{a_j \in \acts \\ \text{for all } j\in \overline{c}}}
q \cdot 
\prod_{i\in c} p_{i,a,c}\cdot 
\prod_{j\in \overline c} p_{j,a_j,c}$ where the quantities
$q$,  and $p_{i,a,c}$, and $p_{j,a_j,c}$  are defined below.
We pick the smallest $i_{\min} \in c$ and define the element $l_c^{k} \in \lstates^k$ for 
$k\in [N]$,  
where $l_c^{k} = l_1^{k}$ if $k \in c$ and $l_c^{k} = l_1^{i_{min},k}$ otherwise.

Then,  we define

$\begin{array}{lll}q &:=& 
\picomm(o^1_1,\ldots,o^N_1)(c) \cdot \\&&
\piact(\langle\langle o_1^1, l _c^{1}\rangle,\ldots\langle o_1^N, l _c^{N}\rangle\rangle)(a) \cdot \\&&
\prod_{j \in \overline{c}} \piact(\langle\langle o_1^1, l^{j,1}\rangle,\ldots\langle o_1^N, l^{j,N}\rangle\rangle)(a_j).\end{array}$

For $i \in c$ and $a = \langle a^1,\ldots,a^N\rangle \in \acts$, let 
 $p_{i,a,c} :=
P^i(\langle o^i_1,l^i_1\rangle,a^i, \langle o^i_2,l_2^i\rangle)
\prod_{k \in \overline i}\frac{P^k(\langle o^k_1,l_c^{k}\rangle,a^k,\langle o^k_2,l_2^{i,k}\rangle)}
{\sum_{l \in \lstates^k}P^k(\langle o^k_1,l_c^{k}\rangle,a^k,\langle o^k_2,l\rangle)}$.

For $j \in \overline c$  and $a_j = \langle a^1,\ldots,a^N\rangle \in \acts$, let 
 $p_{j,a_j,c} := 
P^j(\langle o^j_1,l^j_1\rangle,a^j,\langle o^j_2,l_2^j\rangle)
\prod_{k \in \overline j}\frac{P^k(\langle o^k_1,l^{j,k}\rangle,a^k,\langle o^k_2,l_2^{j,k}\rangle)}
{\sum_{l \in \lstates^j}P^k(\langle o^k_1,l^{j,k}\rangle,a^k,\langle o^k_2,l\rangle)}$.

\item $\widehat s_\init = \langle \widehat{s}_\init^1,\ldots,\widehat{s}_\init^N\rangle$, where
$s_\init^i = \langle o_\init^i, l_\init^i\rangle$ and 
$\widehat{s}_\init^i = \langle\langle o_\init^i, l_\init^i\rangle,\langle l_\init^j\rangle_{j \in \overline i}\rangle$.
\end{itemize}

For $U \subseteq \states$, we define $\widehat U \subseteq\widehat S$ as $\widehat U = \{\langle\langle s^i, m^i \rangle\rangle_{i \in [N]} \in \widehat S\mid \langle s^i \rangle_{i \in [N]} \in U\}$.
We thus lift  $\starget$ and $\savoid$ to  $\wstarget$ and $\wsavoid$. 
For a pair \(\pi = (\picomm,\piact) \in \Picommpos(\ostates,K)\times\Piactpos(M)\) of communication and action policies,   $\prob_{\widehat M_\pi}((\neg\wsavoid) \mathcal{U} \wstarget)$ is the probability of reaching $\starget$ while avoiding $\savoid$, executing $\picomm$ and $\piact$.

Our goal is to compute a pair \((\picomm^*,\piact^*)\) of positional communication and action policies such that $\piact^*$ is optimal for $M$ under unrestricted communication, and among all pairs with optimal action policies,  $\pi^*$ is optimal for $\widehat M$.  

\noindent
\textbf{Problem 1} Given a cooperative Markov game
 \(\widehat{M}=
 (M= (N,\states,\acts,P,s_\init),\ostates^1,\ldots, \ostates^N, \lstates^1,\ldots \lstates^N,K)\) as  in \Cref{def:game} and 
 a reach-avoid objective $(\starget,\savoid)$ for $M$,  find a pair of positional policies  \(\pi^* = (\picomm^*,\piact^*) \in \Picommpos(\ostates,K)\times\Piactpos(M)\) such that
$\prob_{M_{\piact^*}}((\neg\savoid) \mathcal{U} \starget) = v^*(M,\starget,\savoid)$
and for every 
\(\pi = (\picomm,\piact) \in \Picommpos(\ostates,K)\times\Piactpos(M)\) for which we have
$\prob_{M_\piact}((\neg\savoid) \mathcal{U} \starget) = v^*(M,\starget,\savoid)$, it also holds that 
\[\prob_{\widehat M_{\pi^*}}((\neg\wsavoid) \mathcal{U} \wstarget) \geq \prob_{\widehat M_{\pi}}((\neg\wsavoid) \mathcal{U} \wstarget).\]

We restrict  \textbf{Problem 1} to \emph{positional} communication and action policies.  By considering a cooperative Markov game where $|\ostates|=1$  but different sets of agents need to communicate over time,  it is easy to see that communication policies with memory are strictly more powerful.  However,  for the sake of efficient synthesis,  we focus on positional policies.\looseness=-1

\section{Policy Synthesis}\label{sec:synthesis}

For the rest of this section,  we fix a cooperative Markov game
\(\widehat{M}=(M,\ostates^1,\ldots, \ostates^N, \lstates^1,\ldots \lstates^N,K)\) with $M= (N,\states,\acts,P,s_\init)$ as in \Cref{def:game}, and  a reach-avoid objective $(\starget,\savoid) $.
Our goal is to compute a joint action policy that is both optimal and robust to communication restrictions.  To this end, we introduce a cost function based on entropy for information sharing among agents. 
We first present this cost function, followed by our approach for computing 
a pair of action and communication policies.

\subsection{Cost Function for Information Sharing Relative to a Communication Policy}

Recall that we denote 
$\ostates := \ostates^1 \times \ldots \ostates^N$ and 
$\lstates := \lstates^1 \times \ldots \lstates^N$. 
For each communication action $c \in \Acomm$, we define
$\ostates^c := \prod_{i \in c}\ostates^i$ and $\lstates^c := \prod_{i \in c}\lstates^i$.

Let  $\pi = (\picomm,\piact) \in \Picommpos(\ostates,K)\times\Piactpos(M)$ be a pair of positional joint communication and action policies. 

The cost function $D_{(\pi_{\text{comm}},\pi_{\text{act}})}$
 which we define,   measures the information exchange between agents required by the action policy that goes beyond what is allowed by the communication policy.
 It considers all agents $i$ and all possible coalitions $c$ of $K$ agents.  With each, it associates a value  $G^i(\pi_{\text{comm}},\pi_{\text{act}})$  or $G^c(\pi_{\text{comm}},\pi_{\text{act}})$, respectively. 
These are sums of entropy over time,  with additional conditioning on the global observations $\ostates$ and weighted by the probability that at the respective time step the process is ``relevant''. 
\[\begin{array}{lll}
D_{(\pi_{\text{comm}},\pi_{\text{act}})} & := & 
\sum_{i \in [N]} G^i(\pi_{\text{comm}},\pi_{\text{act}}) \\&+ &
\sum_{c \in \Acomm} G^c(\pi_{\text{comm}},\pi_{\text{act}}) \\&-&
\sum_{t=1}^\infty H \left(A_t S_t \middle| S_0 A_1 S_1 \dots A_{t-1}S_{t-1} \right),
\end{array}
\]
where
\[\begin{array}{ll}
G^i(\pi_{\text{comm}},\pi_{\text{act}}) & := 
\sum_{t=1}^{\infty} \sum\limits_{o \in \ostates \atop l^i\in\lstates^i }
w'(o,i) p'(o,l^i)  L'(i,o,l^i)\\
G^c(\pi_{\text{comm}},\pi_{\text{act}}) & := 
\sum_{t=1}^{\infty} \sum\limits_{o \in \ostates \atop l^c\in\lstates^c}
w''(o,c) p''(o,l^c)  L''(c,o,l^c)\\
w'(o, i) & :=\sum\limits_{c \in \Acomm, i\not\in c}\picomm(o)(c)\\
w''(o, c) & : =\picomm(o)(c)\\
p'(o,l^i) & :=  \prob(O_{t-1} = o,L_{t-1}^i = l^i)\\
p''(o,l^c) & :=  \prob(O_{t-1} = o,L_{t-1}^c = l^c)\\
L'(i,o,l^i) & := -\sum\limits_{a^i \in \acts^i, o_1^i \in \ostates^i, l_1^i \in \lstates^i}  p'(a^i,o_1^i,l_1^i,o,l^i) \\
L''(c,o,l^c) & := -\sum\limits_{a^c \in \acts^c, o_1^c \in \ostates^c , l_1^c \in \lstates^c} p''(a^c,o_1^c,l_1^c,o,l^c) 
\end{array}
\]
\[
\begin{array}{l}
p'(a^i,o_1^i,l_1^i,o,l^i) := \\
\prob(A_t^i =a^i,O^i_t=o^i_1,L_t^i=l_1^i \mid O_{t-1}=o, L_{t-1}^i=l^i )\cdot\\
\log\big(\prob(A_t^i =a^i,O_t^i=o_1^i,L_t^i=l_1^i \mid O_{t-1}=o,L_{t-1}^i=l^i )\big);\\
p''(a^c,o_1^c,l_1^c,o,l^c) :=\\
\prob(A_t^c =a^c,O_t^c=o_1^c,L_t^c=l_1^c \mid O_{t-1}=o,L_{t-1}^c=l^c )\cdot\\
\log\big(\prob(A_t^c =a^c,O_t^c=o_1^c,L_t^c=l_1^c \mid O_{t-1}=o, L_{t-1}^c=l^c )\big).
\end{array}
\]

Note that if $K=0$, that is, $\Acomm = \{\emptyset\}$, then $\picomm$ is a constant function and no communication between any agents is allowed.  In such case,  $D_{(\pi_{\text{comm}},\pi_{\text{act}})}$ is the total correlation from~\cite{KarabagNT22}.
When $K > 0$,  only the correlation between agents that are outside of what is allowed by $\picomm$ contributes to the value of $D_{(\pi_{\text{comm}},\pi_{\text{act}})}$. 
If $D_{(\pi_{\text{comm}},\pi_{\text{act}})}$ is $0$,  this means that all dependencies between the agents in $\piact$ are covered by the respective agents being allowed by $\picomm$ to communicate at the necessary points in time.

The cost function $D_{(\pi_{\text{comm}},\pi_{\text{act}})}$ has a key property, namely it allows us to provide an upper bound on the performance loss under restricted communication.  This is established in the next theorem.  The proof can be found in~\Cref{sec:app-synthesis}.

\begin{restatable}{theorem}{performanceloss}\label{thm:performanceloss}
 For any cooperative Markov game \(\widehat{M}\) with MMDP $M$,  reach-avoid objective $(\starget,\savoid)$,  and $\pi = (\picomm,\piact) \in \Picommpos(\ostates,K)\times\Piactpos(M)$,  it holds that
\[\begin{array}{l}
    \prob_{M_{\piact}}((\neg\savoid) \mathcal{U} \starget) - \prob_{\widehat M_{\pi}}((\neg\wsavoid) \mathcal{U} \wstarget) \leq\\
    \sqrt{1-\exp\left(-D_{(\pi_{\text{comm}},\pi_{\text{act}})}  \right)}.
\end{array}
\]
\end{restatable}

Due to the form of $D_{(\pi_{\text{comm}},\pi_{\text{act}})}$, in our method for computing a pair of communication and action policies, described in the rest of the section, we will use a proxy function.
 
\subsection{Policy Synthesis}

Our approach proceeds in two steps. 

\subsubsection{Optimistic Optimal Value for Reach-Avoid Probability}
As we require the action policy to be optimal under unrestricted communication,  in the first step, we compute the optimal value $v^*(M,\starget,\savoid)$ for the reach-avoid probability assuming unrestricted communication. 

\subsubsection{Minimizing the Cost of Communication}
In the second step, we use the value $v^*(M,\starget,\savoid)$ as a threshold in the computation of a pair $(\picomm,\piact)$ of policies.
This threshold constrains the action policy $\piact$ to be optimal under unrestricted communication.
Additionally,  we formulate an objective function based on the cost function $D_{(\pi_{\text{comm}},\pi_{\text{act}})}$.
To this end,  we provide a proxy to $D_{(\pi_{\text{comm}},\pi_{\text{act}})}$,  expressed in terms of occupancy measures.
%
The term $\sum_{t=1}^\infty H \left(A_t S_t \middle| S_0 A_1 S_1 \dots A_{t-1}S_{t-1} \right)$ can be expressed in terms of occupancy measure using existing results~\cite{BiondiLNW14}.
For the other two terms in $D_{(\pi_{\text{comm}},\pi_{\text{act}})}$, we provide upper bounds.

Consider a pair of joint policies \(\pi = (\pi_{comm},\pi_{act})\) and the  Markov chain $M_{\piact}$ induced from the MMDP $M$.
This Markov chain generates a stationary process \(X\), which is the joint path of the agents. 
The entropy \(H\left(X\right)\) of \(X\) has a closed form expression in terms of \(\nu_{s,a}\).

\begin{restatable}{proposition}{expressionentropy}\label{prop:expressH}
The entropy of the joint state--action process until reaching the target can be expressed in terms of the state-action occupancy measure \(\nu_{s,a}\) as
$$
\begin{array}{l}
H(S_0) + 
\sum_{t=1}^{\infty} 
H \left(A_t  S_t  \middle| S_0 A_1 S_1 \dots A_{t-1} S_{t-1}  \right) =\\ 
- \left(\sum_{s,a'}  \nu_{s,a'}  \cdot \log \left( \frac{\nu_{s,a'}}{\sum_{b}{\nu_{s,b}}} \right)\right)\\
- \left(\sum_{s,a',s'}  \nu_{s,a'}  \cdot P(s,a',s') \cdot  \log P(s,a',s')\right).
\end{array}
$$ 
\end{restatable}
	
The path of a single agent $i$ or a group of agents $c$ follows a hidden Markov model where \(X\) is the underlying process and \(X^i\), or  \(X^c\), respectively, is the observed process. 
Therefore, the terms $G^i(\pi_{\text{comm}},\pi_{\text{act}})$ and 
$G^c(\pi_{\text{comm}},\pi_{\text{act}})$ do not have  closed-form expressions based on occupancy measures. 
Instead, we employ stationary processes which induce the same occupancy measures,  and derive expressions that are  upper bounds for $G^i(\pi_{\text{comm}},\pi_{\text{act}})$ and 
$G^c(\pi_{\text{comm}},\pi_{\text{act}})$.

As $s^i = \langle o^i, l^i\rangle$ for some $o^i\in\ostates^i$, $l^i \in \lstates^i$,  we write $\nu_{o^i,l^i,a^i}$ instead of 
$\nu_{s^i,a^i}$. 
For $o \in \ostates$ and $c \in \Acomm$, we define 
$\nu_{o,c} := \nu_{o} \cdot \picomm(o)(c) = (\sum_{l \in \lstates, a\in \acts} \nu_{o,l,a}) \picomm(o)(c)$.

For each agent $i \in [N]$ and each set of agents $c \in \Acomm$, we consider the stationary process that induces the same occupancy measures $\nu_{o,l^i,a^i}$ and  $\nu_{o,l^c,a^c}$, respectively,  as the joint policy. 
We establish the following proposition.

\begin{restatable}{proposition}{expressionsgs}\label{prop:approxG}
Let
$$
\begin{array}{l}
\Bar{G^i} (\pi_{\text{comm}},\pi_{\text{act}}) = \\
- \left(\sum_{o,l^i,a^i}  \nu_{o,l^i,a^i}\cdot  w'(o,i) \cdot \log \left( \frac{\nu_{o,l^i,a^i}}{\sum_{b^i}{\nu_{o,l^i,b^i}}} \right)\right) \\
- \left(\sum_{o,l^i,a^i,o_1^i,l^i_1}  \nu_{o,l^i,a^i} \cdot  w'(o,i)  \cdot
h'(o^i,l^i,a^i, o_1^i,l^i_1) \right),\\
 \Bar{G^c} (\pi_{\text{comm}},\pi_{\text{act}}) = \\
- \left(\sum_{o,l^c,a^c}  \nu_{o,l^c,a^c}  \cdot  w''(o,c) \cdot \log \left( \frac{\nu_{o,l^c,a^c}}{\sum_{b^c}{\nu_{o,l^c,b^c}}} \right)\right)\\
- \left(\sum_{o,l^c,a^c,o_1^c,l^c_1}  \nu_{o,l^c,a^c} \cdot  w''(o,c) \cdot
h''(o^c,l^c,a^c,o_1^c,l^c_1) \right),
\end{array}
$$
$
\begin{array}{l}
w'(o,i) = 
\sum\nolimits_{c \in \Acomm, i\not\in c}\frac{\nu_{o,c}}{\sum\nolimits_{c' \in \Acomm}\nu_{o,c'}},\\
w''(o,c) = 
\frac{\nu_{o,c}}{\sum\nolimits_{c' \in \Acomm}\nu_{o,c'}},
\end{array}
$\\
$
\begin{array}{l}
h'(o^i,l^i,a^i, o_1^i,l^i_1) := \\
P^i(o^i,l^i,a^i)(o_1^i,l^i_1) \cdot  \log P^i(o^i,l^i,a^i)(o_1^i,l^i_1),\\
h''(o^c,l^c,a^c,o_1^c,l^c_1):= \\ 
P^c(o^c,l^c,a^c)(o_1^c,l^c_1) \cdot  \log P^c(o^c,l^c,a^c)(o_1^c,l^c_1),\\
P^c(\langle o^j\rangle_{j\in c},\langle l^j\rangle_{j\in c},\langle a^j\rangle_{j\in c})(\langle o_1^j\rangle_{j\in c},\langle l_1^j\rangle_{j\in c})=\\ \Pi_{j \in c}P^j( o^j, l^j,a^j)( o_1^j, l_1^j).
\end{array}
$

Then,  it holds that $G^i(\pi_{\text{comm}},\pi_{\text{act}}) \leq \Bar{G^i} (\pi_{\text{comm}},\pi_{\text{act}})$ and 
$G^c(\pi_{\text{comm}},\pi_{\text{act}}) \leq \Bar{G^c} (\pi_{\text{comm}},\pi_{\text{act}})$.

 \end{restatable}

Combining \Cref{prop:expressH} and \Cref{prop:approxG}, we obtain an upper bound $\Bar{D}_{(\pi_{\text{comm}},\pi_{\text{act}})}$ on \(D_{(\pi_{\text{comm}},\pi_{\text{act}})}\) based on occupancy measures. That is, we have $D_{(\pi_{\text{comm}},\pi_{\text{act}})} \leq\Bar{D}_{(\pi_{\text{comm}},\pi_{\text{act}})}$ for
\begin{equation}\label{eq:approxcost}
\begin{array}{ll}
\Bar{D}_{(\pi_{\text{comm}},\pi_{\text{act}})}  &: =   
- H(X) + \sum_{i \in [N]} \bar{G}^i(\pi_{\text{comm}},\pi_{\text{act}}) \\&
+ \sum_{c \in \Acomm} \bar{G}^c(\pi_{\text{comm}},\pi_{\text{act}}) .
\end{array}
\end{equation}

\paragraph{Cost optimization problem}
Using the function $\Bar{D}$ defined in  \cref{eq:approxcost}, we formulate bellow, in (\ref{opt}),  the optimization problem with decision variables \(x_{o,l,a}\) and \(x_{o,c}\), representing the occupancy measures for each joint public state-local state-action triplet \((o,l,a)\) and for each joint public state-communication action pair \((o,c)\), respectively. Through appropriate marginalization, \(x_{o,l^i,a^i}\) and \(x_{o,l^c,a^c}\) represent the public state-local state-action occupancy measures for individual agents and groups of agents, respectively.
For synthesis, we assume that the occupancy measure is finite for all states $s \in \states \setminus (\states_{avoid}\cup \states_{target} )$. 
We add absorbing sink-states and corresponding actions to \(\widehat{M}\),  denoted with \( (o_\alpha,l_\alpha) = \left( (o_\alpha^1,l_\alpha^1), \dots, (o_\alpha^n,l_\alpha^n) \right) \) and \( a_\alpha = \left( a_\alpha^1, \dots, a_\alpha^n \right) \), respectively. These states represent the end of the game concerning the reach-avoid objective, that is, for all \((o,l) \in (\starget \cup \savoid)\) we have \(P((o,l),a_\alpha)((o_\alpha,l_\alpha))=1\).

We consider the optimization problem (\ref{opt}) with the objective (\ref{eq:optproblem}) and constraints (\ref{opt-defgi}) -- (\ref{opt-valid}) given below.
The value $v^*$ computed in the first step is used to constrain from below the reach-avoid probability of the action policy. 
Additionally, the value of  the function $\Bar{D}$ must be minimized.

\begin{subequations}
    \label{opt}
\begin{align}
   & \min_{ (x_{o,l,a}, x_{o,c})}{\bar{d} = -h + \sum\limits_{i \in [N]} g^i + \sum\limits_{c \in \Acomm} g^c} \ \ \ \text{s.t.} \label{eq:optproblem}\\
   & \quad \quad
   g^i = \ldots \; \forall i \in [N] 
   \quad \text{/* encodes $\Bar{G^i}$ */} \label{opt-defgi}\\
   & \quad \quad
   g^c = \ldots \;\forall c \in \Acomm
   \quad \text{/* encodes $\Bar{G^c}$ */}\label{opt-defgc}\\
   & \quad\quad
    h = \ldots 
    \quad \text{/* encodes $H(X)$ */}\label{opt-defgh}\\
   & \quad \quad
   w'(o,i) =  \ldots , w''(o,c) =  \ldots
   \ \ \forall o , i ,c \label{opt-defw}\\
   & \quad \quad
   v^* \leq \ldots  \quad \text{/* reach-avoid probability $v$ */}\label{opt-defv}\\
     & \quad \quad
     \sum\nolimits_{a \in \acts \cup \{a_\alpha\}}{x_{o,l,a}} = \ldots
     \ \ \forall (o,l) \in \states\label{opt-flow}\\
	\nonumber
   & \quad \quad
   x_{o,l, a}  \geq 0, \ \    x_{o_\alpha,l_\alpha,a}  = 0 
    \ \ \forall\left(o, l\right) \in \states, a \in \acts\cup\{a_\alpha\}\\
   \nonumber
   & \quad \quad
   x_{o, c}  \geq 0, \ \ x_{o_\alpha,c}  = 0  
   \ \ \forall o \in \ostates, c \in \Acomm\\
   & \quad \quad \sum\nolimits_{l \in \lstates, a \in \acts}{x_{o,l,a}} = \sum\nolimits_{c \in \Acomm}x_{o,c} \ \ \forall o \in \ostates \label{opt-valid}
\end{align}
\end{subequations}
%

\begin{figure*}[t!]
    \centering
    \begin{subfigure}[t]{0.23\textwidth}
    \centering 
        \includegraphics[scale=0.075]{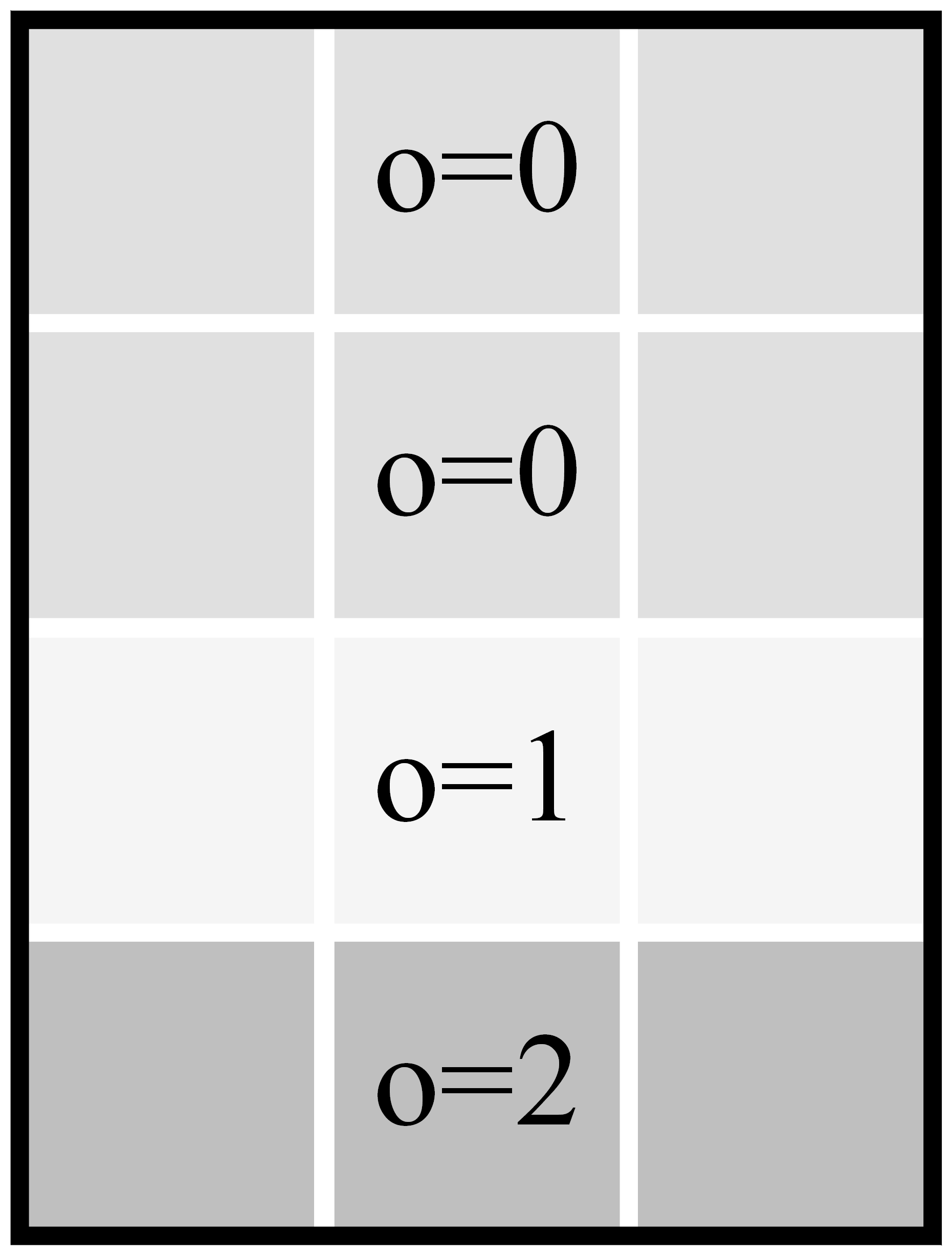}
    \caption{Public information\\labeled regions for Scenario \#1.}
    \label{fig:Env_1_public_region}
    \end{subfigure}
    \hspace{0.3cm}
    \begin{subfigure}[t]{0.23\textwidth}
    \centering 
        \includegraphics[scale=0.075]{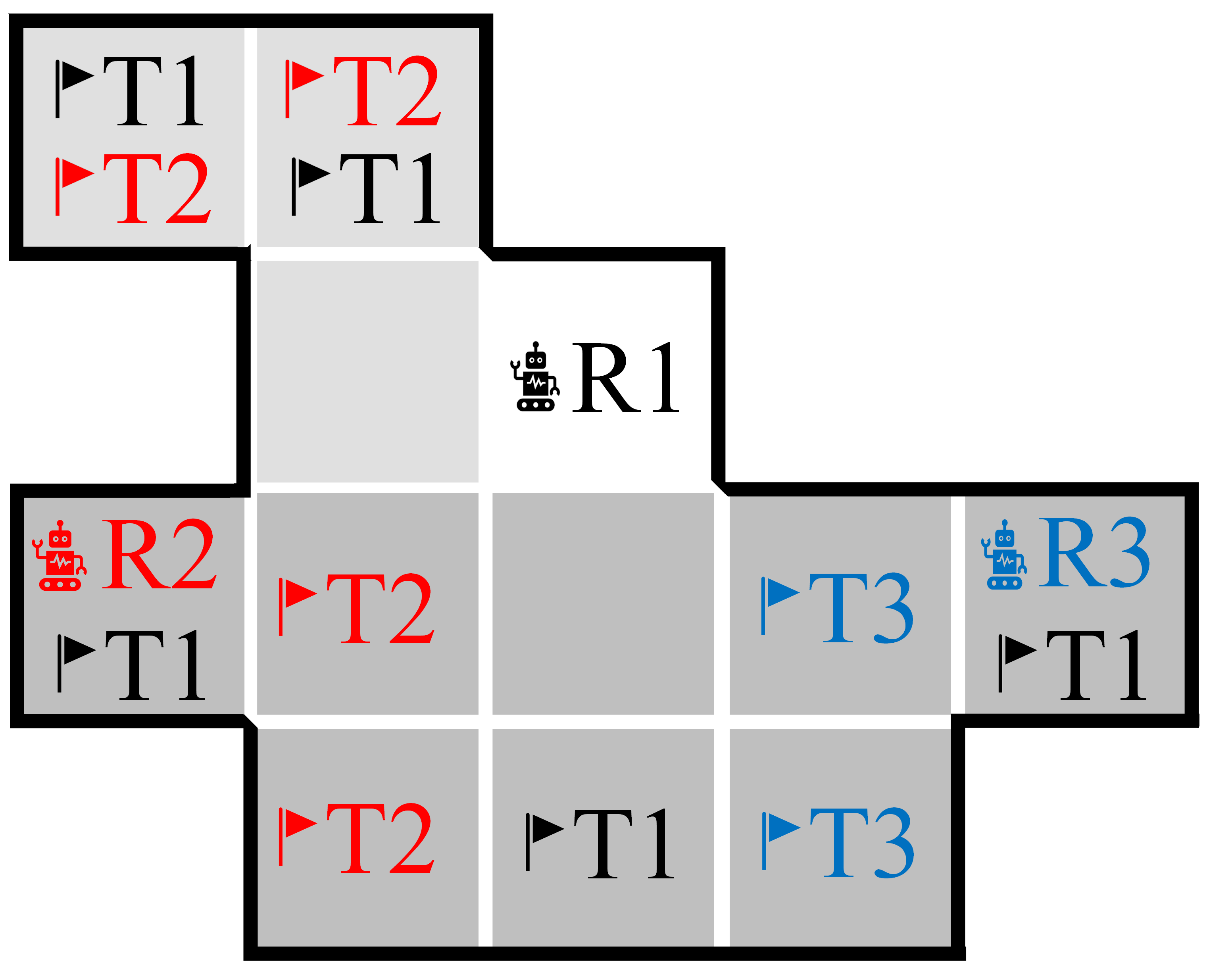}
    \caption{Environment for\\ Scenario \#2.}
    \label{fig:Environment_2}
    \end{subfigure}
    \hspace{0.3cm}
    \begin{subfigure}[t]{0.23\textwidth}
    \centering 
   \includegraphics[scale=0.075]{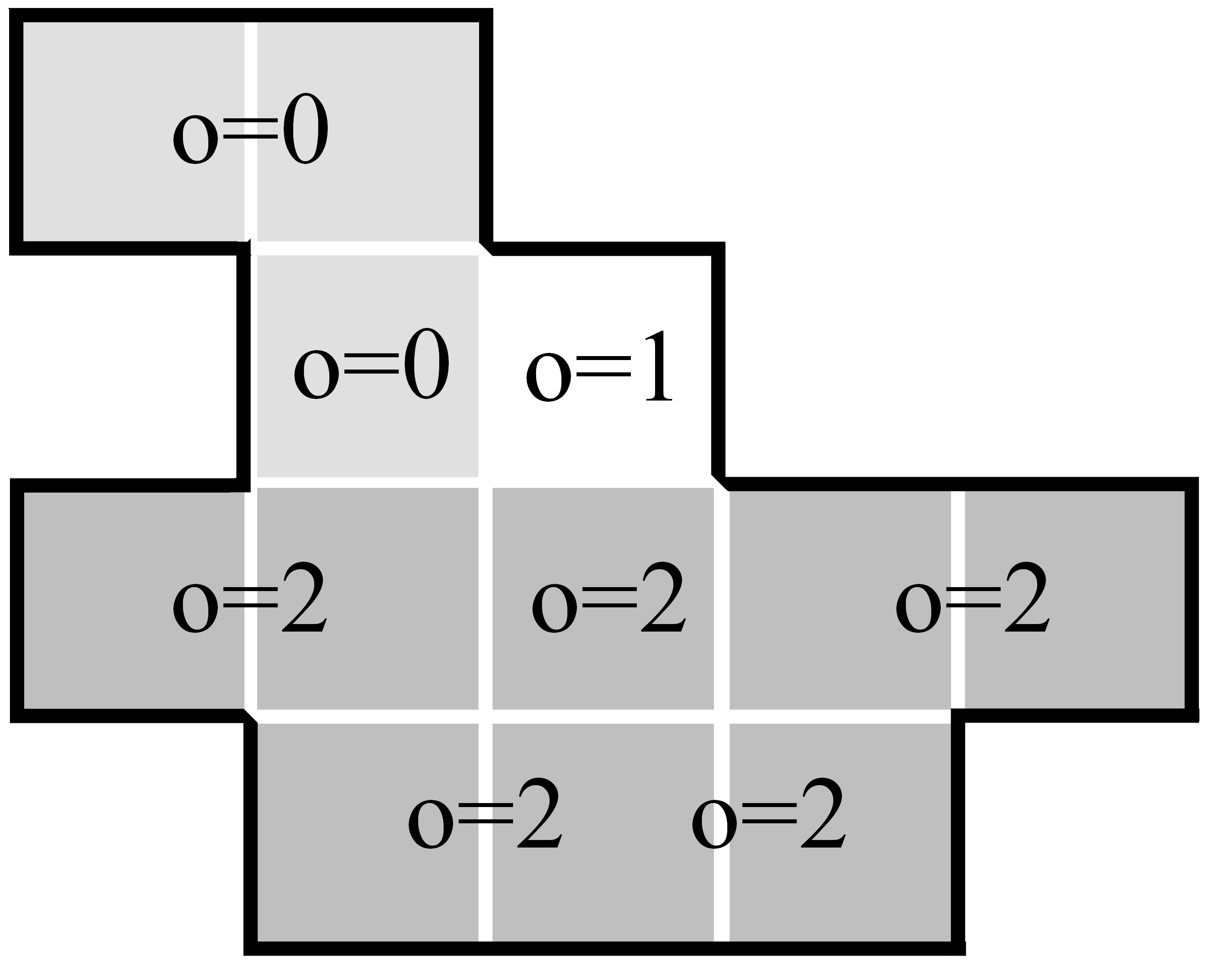}
         \caption{ Public information\\labeled regions for Scenario \#2.}
         \label{fig:Environment_2_region}
    \end{subfigure}
    \vspace{.2cm}

    \begin{subfigure}[t]{0.23\textwidth}
    \centering 
   \includegraphics[scale=0.075]{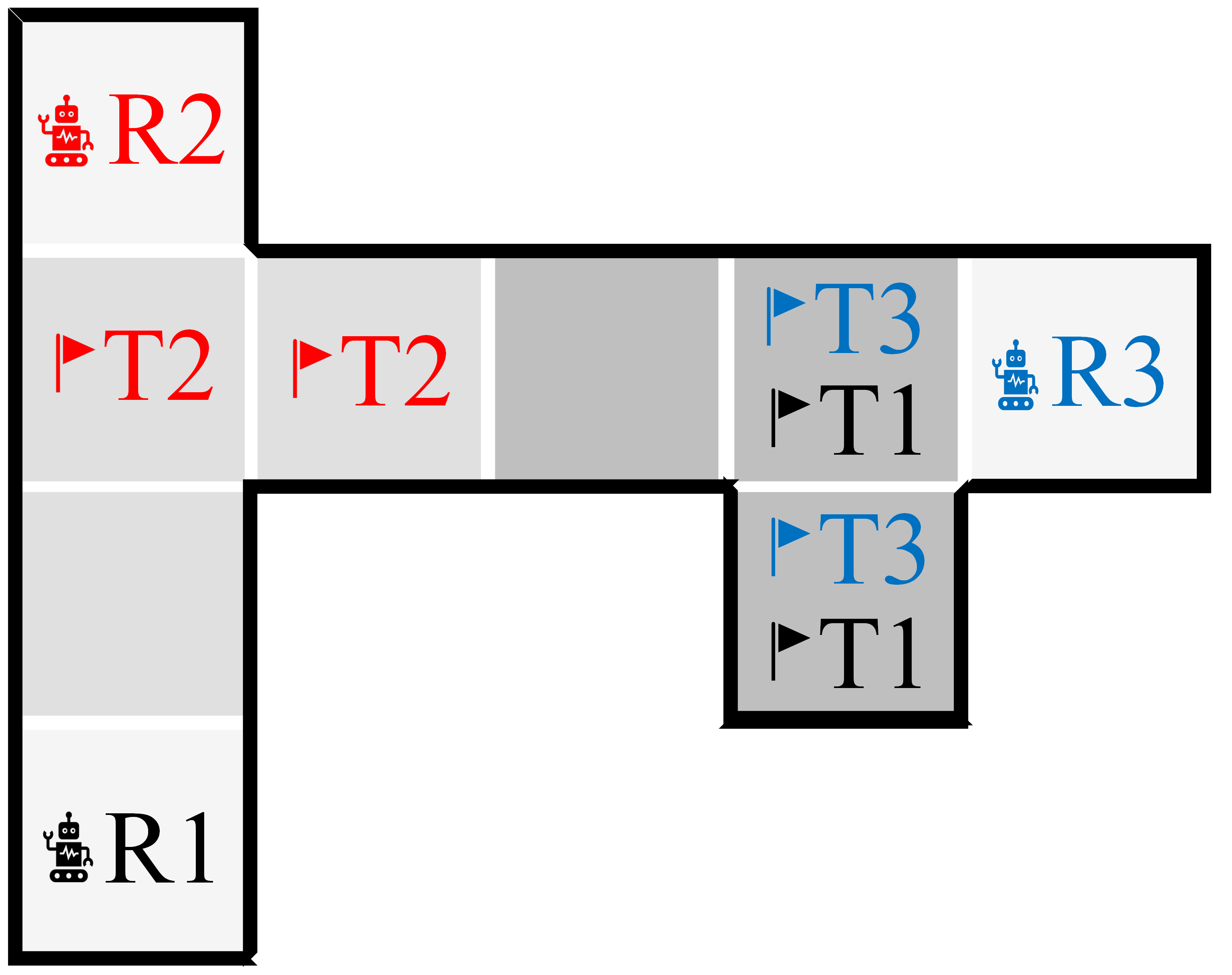}
         \caption{ Environment for\\ Scenario \#3.}
         \label{fig:subfig_map_Env_3}
    \end{subfigure}
    \begin{subfigure}[t]{0.23\textwidth}
    \centering 
        \includegraphics[scale=0.075]{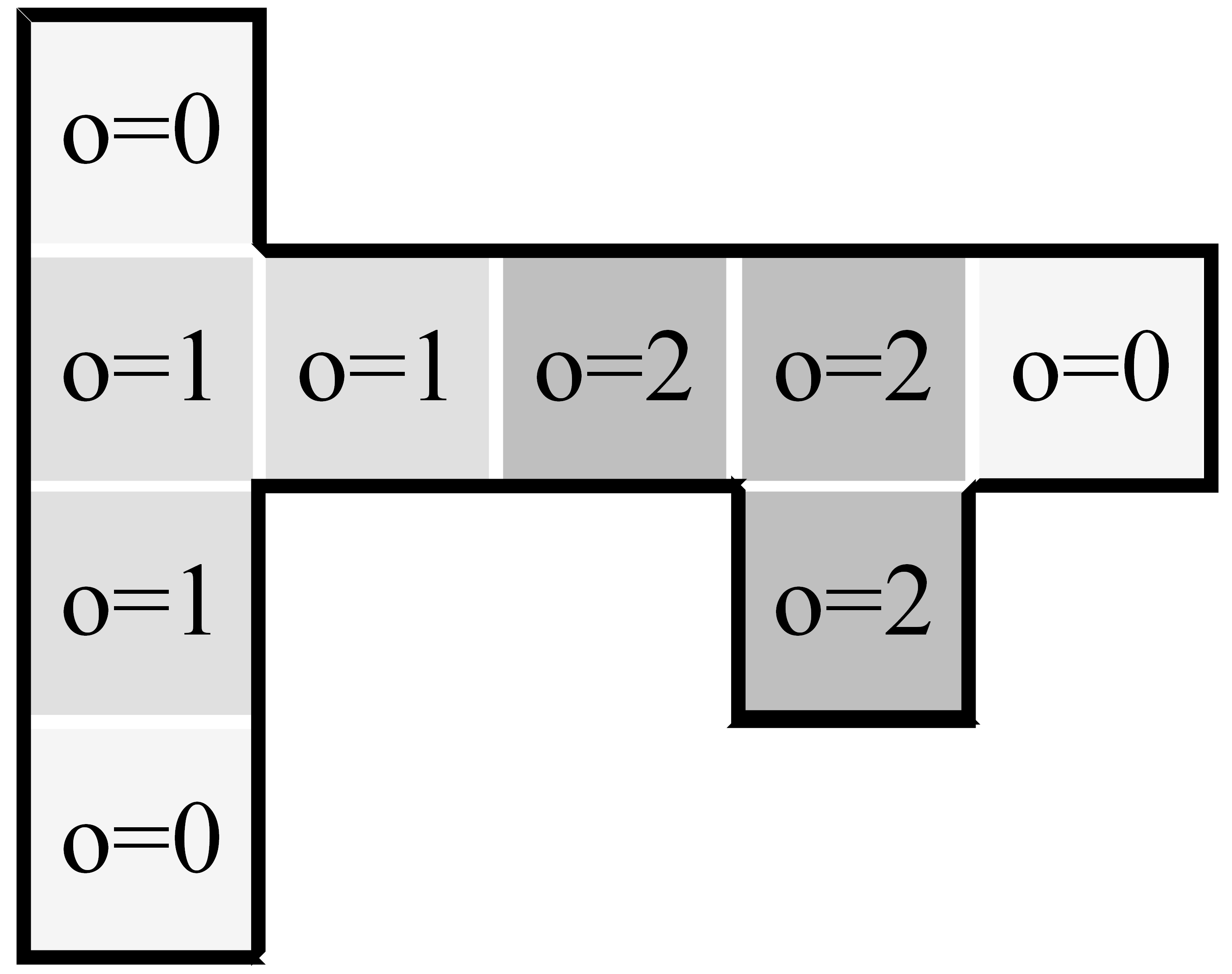}
         \caption{Public information \\labeled regions for Scenario \#3. }
         \label{fig:subfig_region_Env_3}
    \end{subfigure}
    \begin{subfigure}[t]{0.23\textwidth}
    \centering 
	\includegraphics[scale=0.075]{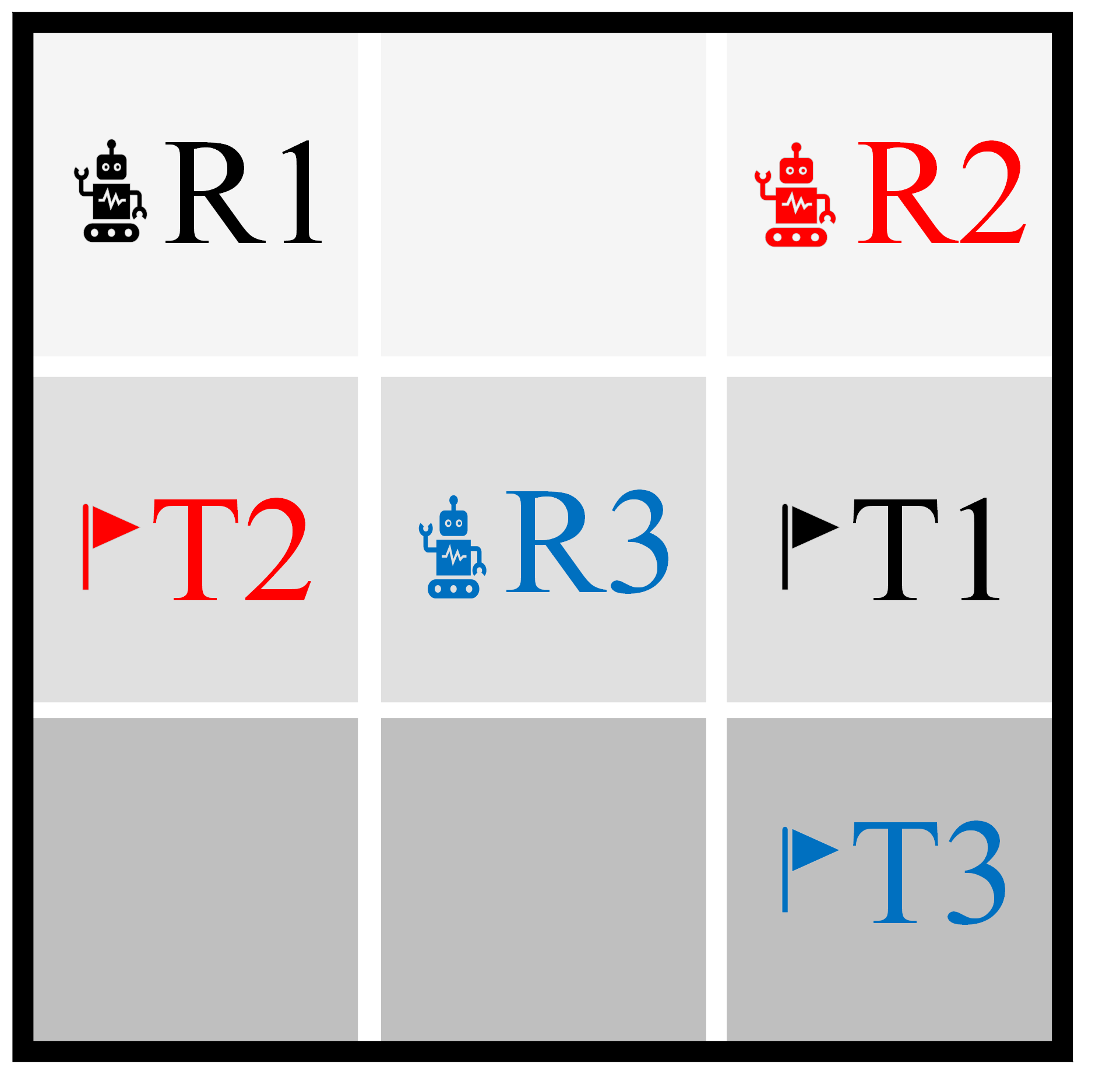}
    \caption{Environment for \\ Scenario \#4.}
    \label{fig:Environment_4}
    \end{subfigure}
    \begin{subfigure}[t]{0.23\textwidth}
    \centering 
        \includegraphics[scale=0.075]{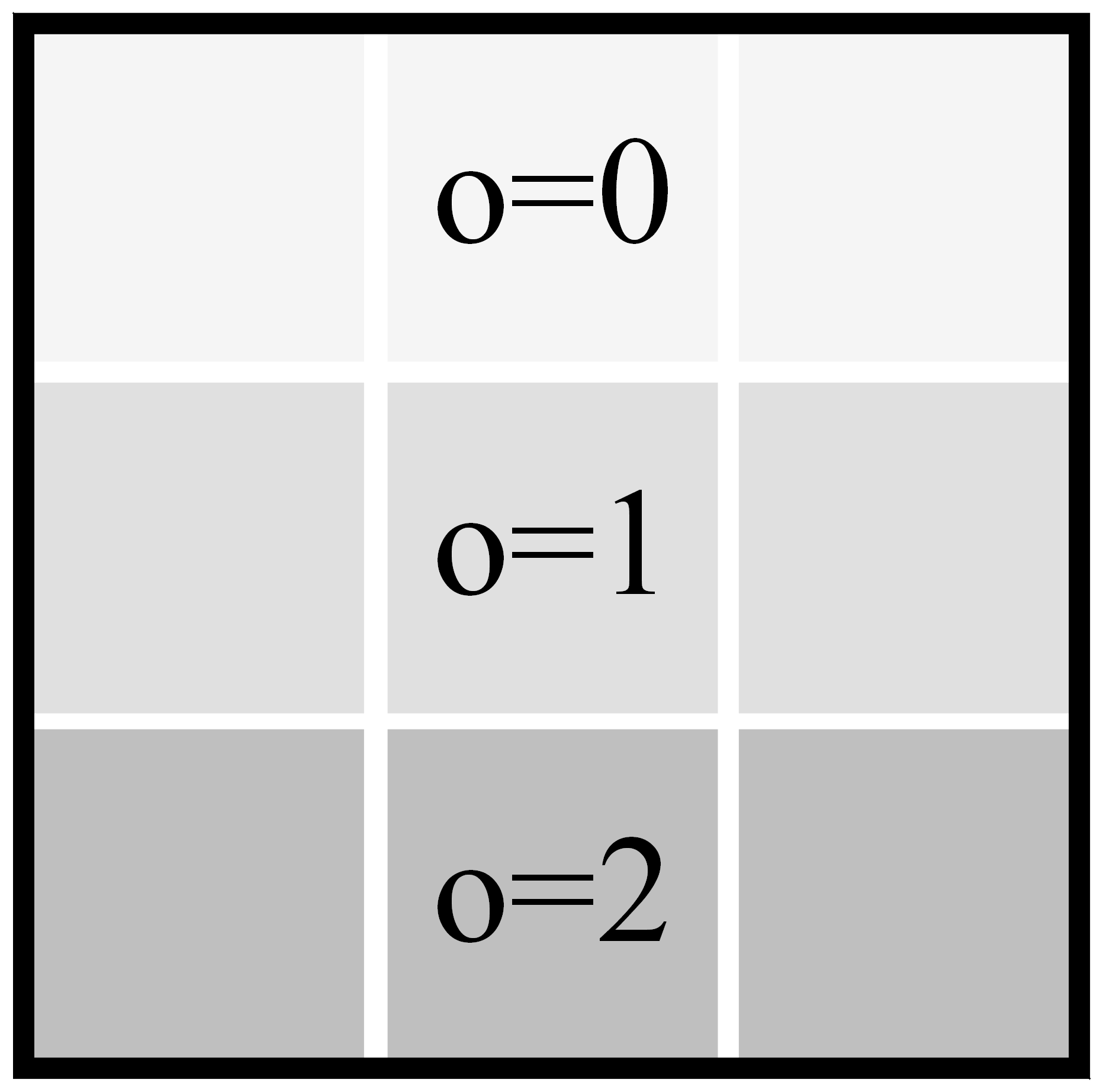}
    \caption{Public information \\labeled regions for Scenario \#4.}
    \label{fig:Environment_4_region}
    \end{subfigure}
    \vspace{-.2cm}
    \caption{Grid environments and regions labeled with public information for the scenarios in \Cref{sec:eval}.  Robots' initial positions are indicated by $R1$, $R2$, and $R3$, and their target positions by $T1$, $T2$, and $T3$.}
\end{figure*}

The constraints (\ref{opt-defgi}) -- (\ref{opt-flow}) are presented below.

Constraints (\ref{opt-defgi}), (\ref{opt-defgc}) and (\ref{opt-defgh}) capture the definitions of the expressions 
$\Bar{G^i} (\pi_{\text{comm}},\pi_{\text{act}}), 
\Bar{G^c} (\pi_{\text{comm}},\pi_{\text{act}})$ and 
\(H\left(X\right)\) from \Cref{prop:approxG} and \Cref{prop:expressH} respectively.  
Formally, 

$
\begin{array}{l}
g^i = -  \Big(\sum_{o,l^i,a^i}  x_{o,l^i,a^i}\cdot  w(o,i) \cdot \log ( \frac{x_{o,l^i,a^i}}{\sum_{b^i}{x_{o,l^i,b^i}}})\Big)\\
\quad - \Big(\sum_{\substack{o,l^i,a^i \\ o_1^i,l^i_1}}  x_{o,l^i,a^i} \cdot  w(o,i)  \cdot P^i(o^i,l^i,a^i)(o_1^i,l^i_1) \\
\quad \cdot  \log P^i(o^i,l^i,a^i)(o_1^i,l^i_1)\Big)\\
g^c = - \Big(\sum_{o,l^c,a^c}  x_{o,l^c,a^c}  \cdot  w(o,c) \cdot \log ( \frac{x_{o,l^c,a^c}}{\sum_{b^c}{x_{o,l^c,b^c}}} )\Big) \\
\quad -\Big(\sum_{\substack{o,l^c,a^c\\o_1^c,l^c_1}}  x_{o,l^c,a^c} \cdot  w(o,c) \cdot P^c(o^c,l^c,a^c)(o_1^c,l^c_1) \\
\quad \cdot  \log P^c(o^c,l^c,a^c)(o_1^c,l^c_1)\Big)\\
h = - \Big(\sum_{s,a'}  x_{s,a'}  \cdot \log ( \frac{x_{s,a'}}{\sum_{b}{x_{s,b}}} )\big)\\
\quad - \big(\sum_{s,a',s'}  x_{s,a'}  \cdot P(s,a')(s') \cdot  \log P(s,a')(s')\Big).
\end{array}$

Similarly,  constraints (\ref{opt-defw}) encode the respective definitions from \Cref{prop:approxG}.  Formally,  we have 

$
\begin{array}{l}
w(o,i) =  \sum\limits_{c \in \Acomm, i\not\in c}\frac{x_{o,c}}{\sum\limits_{c' \in \Acomm}x_{o,c'}}\\
w(o,c) =  \frac{x_{o,c}}{\sum\limits_{c' \in \Acomm}x_{o,c'}}.
\end{array}
$

Constraint (\ref{opt-defv}) lower-bounds the reach-avoid probability for the sought action policy, and constraint~(\ref{opt-flow}) enforces the usual flow constraint for occupancy measures:
$
\begin{array}{l}
v^* \leq \displaystyle\sum_{\substack{(o,l) \in \states \setminus (\states_{avoid} \cup \states_{target}) \\ a \in \acts ,  (o',l') \in \states_{target}}}
x_{o,l,a} \, P(o,l,a)(o',l'), \\ \negthickspace
\displaystyle\sum_{a \in \acts \cup \{a_\alpha\}}\negthickspace\negthickspace{{x_{o,l,a}} = \negthickspace \sum_{\substack{(o',l')\in \states \\ b \in \acts \cup \{a_\alpha\} }} \negthickspace\negthickspace{x_{o',l',b}P(o',l',b)(o,l)}+{\mathbb{1}}_{\{s_\init=s\}}}.
\end{array}
$

Finally,  constraints (\ref{opt-valid}) express the relationship between the occupancy measures \(x_{o,l,a}\) and \(x_{o,c}\)  via the occupancy measure of $o \in \ostates$.

\paragraph{Policies from an optimal solution}
Let   (\(x_{o,l,a}^*\), \(x_{o,c}^*\)) be an optimal solution to the optimization problem (\ref{opt}). We define the pair $(\picomm^*,\piact^*)$ of policies by
\begin{equation}\label{eq:policydef}
\begin{split}
    \picomm^*(o)(c) & = \frac{x_{o,c}^*}{\sum_{d \in \Acomm}{x_{o,d}^*}}, \\
    \piact^*(o,l)(a) & = \frac{x_{o,l,a}^*}{\sum_{b \in \acts}{x_{o,l,b}^*}}. 
    \end{split}
\end{equation}

The next theorem states that $(\picomm^*,\piact^*)$  has the desired properties, namely $\piact^*$ ensures $v^*$ under unrestricted communication, and $(\picomm^*,\piact^*)$ minimizes the value of $\bar{D}$.
 
\begin{restatable}{theorem}{thmoptimization}
 Let \(\widehat{M}\) be a cooperative Markov game with MMDP $M$,  
 $(\starget,\savoid)$ be a reach-avoid objective, 
  and $(\picomm^*,\piact^*)$ be the pair of policies defined by (\ref{eq:policydef}) for an optimal solution to (\ref{eq:optproblem}). 
 Then,   we have
 $\prob_{M_{\piact^*}}((\neg\savoid )\mathcal{U} \starget) = 
 \sup_{\policy} \prob_{M_\policy} ((\neg\savoid) \mathcal{U} \starget).$
 Furthermore, for every  pair $(\picomm',\piact') \in  \Picommpos(\ostates,K)\times\Piactpos(M)$ of positional joint communication and action policies, if the policy $\piact'$ is optimal, that is, if
$\prob_{{M}_{\piact'}}((\neg\savoid )\mathcal{U} \starget) = \sup_{\policy} \prob_{{M}_\policy} ((\neg\savoid) \mathcal{U} \starget)$,  then for function $\bar{D}$ defined in \cref{eq:approxcost} we have
 \( \bar{D}_{(\picomm^*, \piact^*)} \leq  \bar{D}_{(\picomm',\piact')}  \).
\end{restatable}

\section{Experimental Evaluation}\label{sec:eval}
We evaluated our approach on four multi-robot navigation scenarios.  In each scenario, we consider a MAS with three agents and $K = 2$. We used the formalization in \Cref{sec:synthesis} to synthesize for each problem a pair of communication and action policies.
All experiments were performed on a Macbook Pro with an Apple M2 chip and 32GB memory. 
The two optimization problems are solved using GLPK \cite{GLPK} and SNOPT \cite{snopt},  respectively.  The detailed setup and results of all scenarios are included in~\Cref{sec:app-benchmarks}.

With this evaluation we aim to demonstrate the following.
\begin{enumerate}
\item Our method synthesizes policies with zero communication cost (which means they fully conform to the communication restrictions) when such policies exist.
\item Our method is capable of synthesizing policies that adhere to communication restriction while incurring zero communication cost—an outcome that cannot be achieved through approaches based solely on minimizing total correlation as an objective function.
\item Our method can synthesize communication policies that adapt dynamically to the current state of public information so that the set of communicating agents changes.
\item There is a trade-off between performance (the reach-avoid probability) and the value of our cost function.
\end{enumerate}

Next, we describe the four scenarios and the respective results.  Further details can be found in~\Cref{sec:app-benchmarks}.

\subsubsection{Scenario \#1 with Navigation Tasks}
We consider the environment in \Cref{fig:Env_1_figure}.
The task of the robots,  initialized as marked in the figure,  is to navigate to one of their target cells,  labeled \(T1\), \(T2\), and \(T3\), respectively, avoiding collision.
At any given time, only two robots can communicate and share precise locations.  Which ones communicate is decided based on public information,  which is the current regions of the robots. The  regions,  labeled \(o=0\), \(o=1\), and \(o=2\) are shown in \Cref{fig:Env_1_public_region}.
The possible actions are moving in one of the four cardinal directions  or remaining in place. 
Moves succeed with probability \(0.9\) and fail, resulting in remaining in the current cell,  with probability \(0.1\).  The remain action and impossible actions stay in place.

\paragraph{Result} 
The method proposed in Section~\ref{sec:synthesis} generates a pair of policies where the action policy is optimal under unrestricted communication,  with reach-avoid probability \(0.99\).
The synthesized communication policy results in \emph{cost zero}.
It assigns probability $1$ to robots \(R1\) and \(R2\)  communicating. The synthesized action policy matches that,  unlike other possible action policies with reach-avoid probability \(0.99\).
Thus, our method identifies a suitable action policy that can be equipped with a communication policy achieving cost zero.

\subsubsection{Scenario \#2 with a Swarm Intersection}
The goal of this scenario is to compare our method with the approach based on minimizing the total correlation. In the scenario depicted in Figure~\ref{fig:Environment_2}, three robots are required to navigate to their respective target locations while avoiding collisions. The publicly available information corresponding to each region is illustrated in Figure~\ref{fig:Environment_2_region}. Again, only two out of three agents can communicate their precise location without an extra communication cost.

The work closest to ours is~\cite{KarabagNT22}, which employs total correlation to measure dependencies between agents and to synthesize an (action) policy that is robust to communication loss. 
However,  \cite{KarabagNT22} does not consider communication policies since robustness is w.r.t.~complete loss of communication, not to partial limitation of the communication.
Thus,  there is a baseline for comparison only for the action policies synthesized by our approach (which are optimal),  namely the policies computed using total correlation in the objective function.
Minimizing total correlation cannot yield action policies that satisfy communication restrictions, whereas our method can do so, and also generates a matching communication policy.

\paragraph{Result}
In scenario \#2, our approach finds a pair of action and communication policies that satisfy the reach-avoid property with probability \(1\) and ensure communication cost \(0\). Thus,  it finds policies that satisfy the communication restriction.
In contrast,  the action policy generated using total correlation violates the communication restriction at time $t=2$ when it requires coordination among three agents. 
Here, the minimum total correlation is \(0.591\), while the total correlation of the policy computed by our method is \(0.693\), and thus it will not be computed by minimizing total correlation.

\subsubsection{Scenario \#3 with a Hallway}
In this scenario, we consider the environment in \Cref{fig:subfig_map_Env_3}. 
The public information regions, labeled \(o=0\), \(o=1\), and \(o=2\), are depicted in \Cref{fig:subfig_region_Env_3}.
In this scenario,  coordination among certain robots is critical at certain time steps to avoid collisions.  As before, only two of the robots are allowed to communicate at any given time.

\paragraph{Result}
The synthesized action policy achieves a reach-avoid probability of \(1\), and the communication policy ensures zero communication cost by selecting the appropriate set of robots to communicate in different public information states.

\subsubsection{Scenario \#4 with High Uncertainty}
The robots in the environment shown in Figure~\ref{fig:Environment_4} must navigate to their respective target cells while avoiding collisions.  The public information regions here correspond to the rows, as shown in Figure~\ref{fig:Environment_4_region}.
The actions are the same as in Scenario \#1. 
Here, however, the move actions lead to the desired cell with probability $0.9$,  and the remaining \(0.1\) probability is redistributed across the current cell and all other neighboring cells.  For impossible move actions,  the full transition probability \(1.0\) is redistributed among the current cell and all valid neighboring cells. As in previous scenarios, when communication is restricted, only two of three agents can communicate. In this scenario, a crowded intersection necessitates communication among all agents to  avoid collisions.

\paragraph{Result}
Here, the optimal reach-avoid probability under full communication is \(0.958\). 
No pair of action and communication policies exists that achieves this probability with zero communication cost.
If, however, we lower the threshold and only ask for a policy that guarantees \(0.92\) reach-avoid probability under unrestricted communication,  then our method synthesizes one with zero cost.

\subsubsection{Performance}
Our evaluation is focused on evaluating the quality of our approach on a range of relatively small but interesting problem instances, as well as its principle feasibility. 
Although our method requires solving large nonlinear optimization problems to synthesize a pair of policies, the running time remains reasonable for the considered benchmarks. 
The runtimes range from a few minutes for Environment~\#3 (with 528 constraints and 62,581 variables) to one hour for Environment~\#4 (with 531 constraints and 62,956 variables) and Environment~\#2 (with 1,332 constraints and 163,081 variables), and up to three hours for Environment~\#1 (with 1,344 constraints and 164,581 variables).
Clearly,  the performance depends on the number of agents and the sparsity of the transition probability matrices.  In the future, we plan to conduct larger-scale experiments and develop techniques for further improving performance.

\section{Discussions and Future Work}\label{sec:discussion}

\paragraph{Limitations and Extensions}
One limitation of the model we study in this paper is that the number $K$ of agents always allowed to communicate is fixed.  We can incorporate a dynamically changing $K$ as part of the state,  making the model and cost function more complex.  Another limitation is the restriction to positional policies, which allows us to formulate the problem via occupancy measures.  In the future we will consider extensions with bounded memory communication policies.
Another challenge is the scalability, in particular with growing number of agents,  which we plan to address by developing methods that iteratively improve the policies for subsets of agents.
Additionally,  we assume that the model, including public information is given. While in many cases identifying public information is natural (regions, fixed capabilities),  exploring relations to agents' observations and approaches to deriving such information is an interesting direction.

\paragraph{Richer Communication Models} The focus of our work is to establish a rigorous theoretical foundation, and provide novel insights and methodology. 
Our main aim is to provide the necessary basis for further theoretical exploration and practical applications. 
Moving forward, we plan to study extensions with richer communication restrictions (such as dynamic changes in communication availability) and explore ways to make the approach robust to implementation aspects such as delayed or noisy communication.

\paragraph{Identifying Public Information} 
The granularity of the public information affects the number of decision variables (for the communication policy) and hence the performance of policy synthesis.   More (i.e.,  finer) public information leads to higher computation times. 
The choice of public information depends on the application and available communication bandwidth. In many cases this is natural, such as coarser geographic regions or agents' capabilities.
Developing techniques for identifying public information is one interesting direction for future work.

\bibliographystyle{named}
\bibliography{main}

\begin{thebibliography}{}

\bibitem[\protect\citeauthoryear{Biondi \bgroup \em et al.\egroup
  }{2014}]{BiondiLNW14}
Fabrizio Biondi, Axel Legay, Bo~Friis Nielsen, and Andrzej Wasowski.
\newblock Maximizing entropy over markov processes.
\newblock {\em J. Log. Algebraic Methods Program.}, 83(5-6):384--399, 2014.

\bibitem[\protect\citeauthoryear{Bretagnolle and
  Huber}{1979}]{Bretagnolle-Huber-inequality}
Jean Bretagnolle and Catherine Huber.
\newblock Estimation des densit{\'e}s: risque minimax.
\newblock {\em Zeitschrift f{\"u}r Wahrscheinlichkeitstheorie und verwandte
  Gebiete}, 47:119--137, 1979.

\bibitem[\protect\citeauthoryear{Cover and Thomas}{2006}]{cover1999elements}
Thomas~M. Cover and Joy~A. Thomas.
\newblock {\em Elements of information theory {(2.} ed.)}.
\newblock Wiley, 2006.

\bibitem[\protect\citeauthoryear{Gill \bgroup \em et al.\egroup }{2018}]{snopt}
Philip~E Gill, Walter Murray, Michael~A Saunders, and Elizabeth Wong.
\newblock User’s guide for snopt 7.7: Software for large-scale nonlinear
  programming.
\newblock {\em Center for Computational Mathematics Report CCoM}, 15(3), 2018.

\bibitem[\protect\citeauthoryear{Goldman and
  Zilberstein}{2004}]{decentralizedcontrolCS}
Claudia~V. Goldman and Shlomo Zilberstein.
\newblock Decentralized control of cooperative systems: Categorization and
  complexity analysis.
\newblock {\em J. Artif. Intell. Res.}, 22:143--174, 2004.

\bibitem[\protect\citeauthoryear{Guestrin \bgroup \em et al.\egroup
  }{2001}]{guestrin2001multiagent}
Carlos Guestrin, Daphne Koller, and Ronald Parr.
\newblock Multiagent planning with factored mdps.
\newblock {\em Advances in neural information processing systems}, 14, 2001.

\bibitem[\protect\citeauthoryear{Karabag \bgroup \em et al.\egroup
  }{2022}]{KarabagNT22}
Mustafa~O. Karabag, Cyrus Neary, and Ufuk Topcu.
\newblock Planning not to talk: Multiagent systems that are robust to
  communication loss.
\newblock In Piotr Faliszewski, Viviana Mascardi, Catherine Pelachaud, and
  Matthew~E. Taylor, editors, {\em 21st International Conference on Autonomous
  Agents and Multiagent Systems, {AAMAS} 2022, Auckland, New Zealand, May 9-13,
  2022}, pages 705--713. International Foundation for Autonomous Agents and
  Multiagent Systems {(IFAAMAS)}, 2022.

\bibitem[\protect\citeauthoryear{Makhorin}{2008}]{GLPK}
Andrew Makhorin.
\newblock Glpk (gnu linear programming kit).
\newblock {\em http://www. gnu. org/s/glpk/glpk. html}, 2008.

\bibitem[\protect\citeauthoryear{Melo and
  Veloso}{2011}]{sparseinteractionsmelo2011decentralized}
Francisco~S Melo and Manuela Veloso.
\newblock Decentralized mdps with sparse interactions.
\newblock {\em Artificial Intelligence}, 175(11):1757--1789, 2011.

\bibitem[\protect\citeauthoryear{Rizk \bgroup \em et al.\egroup
  }{2018}]{rizk2018decision}
Yara Rizk, Mariette Awad, and Edward~W Tunstel.
\newblock Decision making in multiagent systems: A survey.
\newblock {\em IEEE Transactions on Cognitive and Developmental Systems},
  10(3):514--529, 2018.

\bibitem[\protect\citeauthoryear{Shannon and
  Weaver}{1949}]{shannon1949mathematical}
Claude~E Shannon and Warren Weaver.
\newblock The mathematical theory of communication. university of illinois.
\newblock {\em Urbana}, 117:10, 1949.

\bibitem[\protect\citeauthoryear{Wu \bgroup \em et al.\egroup }{2011}]{WuZC11}
Feng Wu, Shlomo Zilberstein, and Xiaoping Chen.
\newblock Online planning for multi-agent systems with bounded communication.
\newblock {\em Artif. Intell.}, 175(2):487--511, 2011.

\end{thebibliography}

\cleardoublepage 

\appendix

\onecolumn  

\section{Proofs from~\Cref{sec:synthesis}}\label{sec:app-synthesis}
 \begin{definition}
     The \emph{Kullback–Leibler divergence (KL divergence)} between two probability mass functions $p(x)$ and $q(x)$ with the same countable support $V$, is defined as $D_{KL}\left(p \parallel q\right)=\sum_{x \in V} p\left(x\right) \log \left(\frac{p\left(x\right)}{q\left(x\right)}\right)$. 
 \end{definition}
 
In the above definition,  $p \log\frac{p}{0}=\infty$. 
The KL divergence is always non-negative and is zero if and only if $p(x)=q(x)$.
The KL divergence quantifies how $p(x)$ differs from $q(x)$. 
In order to show \cref{thm:performanceloss}, we first prove a  lemma that establishes a relationship between the performance loss value and the KL divergence between the distribution of joint paths induced by the joint policy executed without communication restriction and under communication restrictions.  

\begin{lemma}\label{lemma:KL-D}
Let  $\widehat{M}$ be a cooperative Markov game as in \Cref{def:game} and 
 $\pi = (\picomm,\piact) \in \Picommpos(\ostates,K)\times\Piactpos(M)$.
Let $\Gamma_{M_{\piact}}$ be the distribution of joint paths induced by the action policy $\piact$ on $M$, and let $\Gamma_{\widehat{M}_{\pi}}$ be the distribution of joint paths in $M$ induced by $\pi$ on $\widehat{M}$.
Then it holds that 
\begin{align*}
    \prob_{M_{\piact}}((\neg\savoid) \mathcal{U} \starget) - \prob_{\widehat M_{\pi}}((\neg\wsavoid) \mathcal{U} \wstarget) \leq
     \sqrt{1-\exp\left(-D_{KL}\left(\Gamma_{M_\piact} \parallel \Gamma_{\widehat{M}_{\pi}} \right) \right)}. 
\end{align*}
\end{lemma}

\begin{proof}
Let \(T\) denote the set of paths in $M$ reaching \(\starget\), and let \(T'\) be a set of paths in $M$ chosen arbitrarily. Also denote a generic path by $\zeta =s_0 a_1 s_1\ldots$. Then,
\begin{align}
 \prob_{M_{\piact}}((\neg\savoid) \mathcal{U} \starget) - \prob_{\widehat M_{\pi}}((\neg\wsavoid) \mathcal{U} \wstarget) \nonumber &= \sum_{\zeta \in T} \Gamma_{M_{\piact}}\left(\zeta\right) - \Gamma_{\widehat{M}_{\pi}}\left(\zeta\right) 
\nonumber \\
& \leq \left| \sum_{\zeta \in T} \Gamma_{M_{\piact}}\left(\zeta\right) - 
\Gamma_{\widehat{M}_{\pi}}\left(\zeta\right) \right| 
\nonumber\\
& \leq \sup_{T'} \left| \sum_{\zeta \in T'} \Gamma_{M_{\piact}}\left(\zeta\right) - 
\Gamma_{\widehat{M}_{\pi}}\left(\zeta\right) \right|
\nonumber\\
& \leq \sqrt{1-\exp\left(-D_{KL}\left(\Gamma_{M_\piact} \parallel \Gamma_{\widehat{M}_{\pi}} \right) \right)} \, , \label{eq:Bretagnolle-Huber_line4}
\end{align}
where \eqref{eq:Bretagnolle-Huber_line4} is due to Bretagnolle-Huber inequality \cite{Bretagnolle-Huber-inequality}.
\end{proof}

Next, we establish an upper bound on the performance loss based on the cost function we introduced in \Cref{sec:synthesis}, which in turn gives a lower bound on the reach avoid probability under restricted communication.
\performanceloss*
\begin{proof}
Let \(T\) denote the set of paths in $M$ reaching \(\starget\), and let \(T'\) be a set of paths in $M$ chosen arbitrarily. 
We denote a generic path by $\zeta =s_0 a_1 s_1\ldots$. We use \( \mu_{M_\piact}(.)\) and  \(\mu_{\widehat{M}_{\pi}}(.)\) to denote the probability of a state or a path under the distribution of \( \Gamma_{M_\piact}\) and  \(\Gamma_{\widehat{M}_{\pi}}\), respectively. Then,
\begin{align*}
& D_{KL}\left(\Gamma_{M_\piact} \parallel \Gamma_{\widehat{M}_{\pi}}\right) = \sum\limits_{\zeta} 
\mu_{M_\piact}\left(\zeta\right) \cdot \log \left( 
\frac{\mu_{M_\piact}\left(\zeta\right)}
{\mu_{\widehat{M}_{\pi}}\left(\zeta\right)} 
\right) \\
& = \sum\limits_{\zeta} 
\mu_{M_\piact}\left(s_0\right)
\cdot \mu_{M_\piact}\left(a_1 s_1 \mid s_0\right)
\cdot \mu_{M_\piact}\left(a_2 s_2 \mid s_0 a_1 s_1\right) \cdots \cdot \log \left[
\frac{\mu_{M_\piact}\left(s_0\right) 
\mu_{M_\piact}\left(a_1 s_1 | s_0\right) 
\mu_{M_\piact}\left(a_2 s_2 | s_0 a_1 s_1\right) \cdots}
{\mu_{\widehat{M}_{\pi}}\left(s_0\right) 
\mu_{\widehat{M}_{\pi}}\left(a_1 s_1 | s_0\right) 
\mu_{\widehat{M}_{\pi}}\left(a_2 s_2 | s_0 a_1 s_1\right) \cdots } 
\right] \\
& = \sum\limits_{t=1}^{\infty} \sum\limits_{\zeta} 
\mu_{M_\piact}\left(s_0\right)
\cdot \mu_{M_\piact}\left(a_1 s_1 \mid s_0\right) \cdots \cdot \mu_{M_\piact}\left(a_t s_t \mid s_0 a_1 s_1 \ldots a_{t-1} s_{t-1} \right) \cdots \cdot \log \left( 
\frac{\mu_{M_\piact}\left(a_t s_t \mid s_0 a_1 s_1 \ldots a_{t-1} s_{t-1} \right)}
{\mu_{\widehat{M}_{\pi}}\left(a_t s_t \mid s_0 a_1 s_1 \ldots a_{t-1} s_{t-1} \right)} 
\right)
\end{align*}

Note that each state \(s\) includes two parts of publicly observable \(o\) and local states \(l\). At each time point \(t\) by the log sum inequality \cite{cover1999elements}, \(\sum\limits_{c \in \Acomm} \picomm(o)(c)=1\), and
\begin{align*}
& \mu_{\widehat{M}_{\pi}}\left(a_t s_t \mid s_0 a_1 s_1 \ldots a_{t-1} s_{t-1} \right) = \sum\limits_{c \in \Acomm} 
\picomm(o_{t-1})(c) \cdot 
\mu^{c}\left(a_t^c s_t^c \mid s_0 a_1 s_1 \ldots a_{t-1} s_{t-1} \right) \cdot \prod_{i \notin c} 
\mu^{i}\left(a_t^i s_t^i \mid s_0 a_1 s_1 \ldots a_{t-1} s_{t-1} \right)
\end{align*}
we have (\ref{eq:long_equation4}) below
\begin{subequations}  \label{eq:long_equation4}
\begin{align} 
\sum\limits_{\zeta} 
& \mu_{M_\piact}(s_0) \mu_{M_\piact}(a_1 s_1 \mid s_0) \cdots \mu_{M_\piact}(a_t s_t \mid s_0 a_1 s_1 \ldots a_{t-1} s_{t-1}) \cdots \nonumber\\ 
& \quad \cdot \log \left( 
\frac{\mu_{M_\piact}(a_t s_t \mid s_0 a_1 s_1 \ldots a_{t-1} s_{t-1})}
{\sum\limits_{c \in \Acomm} \picomm(o_{t-1})(c) 
\mu^{c}(a_t^c s_t^c \mid s_0 a_1 s_1 \ldots a_{t-1} s_{t-1}) \cdot 
\prod_{i \notin c} \mu^{i}(a_t^i s_t^i \mid s_0 a_1 s_1 \ldots a_{t-1} s_{t-1})} 
\right) \\ 
= \sum\limits_{\zeta} 
& \mu_{M_\piact}(s_0) \mu_{M_\piact}(a_1 s_1 \mid s_0) \cdots \left( \sum\limits_{c \in \Acomm} \picomm(o_{t-1})(c) \right) 
\mu_{M_\piact}(a_t s_t \mid s_0 a_1 s_1 \ldots a_{t-1} s_{t-1}) \cdots \nonumber\\ 
& \quad \cdot \log \left( 
\frac{\sum\limits_{c \in \Acomm} \picomm(o_{t-1})(c) 
\mu_{M_\piact}(a_t s_t \mid s_0 a_1 s_1 \ldots a_{t-1} s_{t-1})}
{\sum\limits_{c \in \Acomm} \picomm(o_{t-1})(c) 
\mu^{c}(a_t^c s_t^c \mid s_0 a_1 s_1 \ldots a_{t-1} s_{t-1}) \cdot 
\prod_{i \notin c} \mu^{i}(a_t^i s_t^i \mid s_0 a_1 s_1 \ldots a_{t-1} s_{t-1})} 
\right) \\ 
\leq \sum\limits_{\zeta} 
& \mu_{M_\piact}(s_0) \cdots \mu_{M_\piact}(a_{t-1} s_{t-1} \mid s_0 \ldots a_{t-2} s_{t-2}) \cdot \mu_{M_\piact}(a_{t+1} s_{t+1} \mid s_0 \ldots a_t s_t) \cdots \nonumber\\ 
& \quad \cdot \Biggl( \sum\limits_{c \in \Acomm} \picomm(o_{t-1})(c) 
\mu_{M_\piact}(a_t s_t \mid s_0 a_1 s_1 \ldots a_{t-1} s_{t-1}) \nonumber\\ 
& \quad \cdot \log \left( 
\frac{\picomm(o_{t-1})(c) \mu_{M_\piact}(a_t s_t \mid s_0 a_1 s_1 \ldots a_{t-1} s_{t-1})}
{\picomm(o_{t-1})(c) \mu^{c}(a_t^c s_t^c \mid s_0 a_1 s_1 \ldots a_{t-1} s_{t-1}) \cdot 
\prod_{i \notin c} \mu^{i}(a_t^i s_t^i \mid s_0 a_1 s_1 \ldots a_{t-1} s_{t-1})} 
\right) \Biggr) \\ 
= \sum\limits_{\zeta} & \sum\limits_{c \in \Acomm} 
 \picomm(o_{t-1})(c) \mu_{M_\piact}(\zeta) \nonumber\\ 
& \quad \cdot \log \left( 
\frac{\mu_{M_\piact}(a_t s_t \mid s_0 a_1 s_1 \ldots a_{t-1} s_{t-1})}
{\mu^{c}(a_t^c s_t^c \mid s_0 a_1 s_1 \ldots a_{t-1} s_{t-1}) \cdot 
\prod_{i \notin c} \mu^{i}(a_t^i s_t^i \mid s_0 a_1 s_1 \ldots a_{t-1} s_{t-1})} 
\right) \\ 
= \sum\limits_{\zeta} & \sum\limits_{c \in \Acomm} \picomm(o_{t-1})(c) \mu_{M_\piact}(\zeta) \cdot \log \left( \mu_{M_\piact}(a_t s_t \mid s_0 a_1 s_1 \ldots a_{t-1} s_{t-1}) \right) \nonumber\\ 
- & \sum\limits_{\zeta} \sum\limits_{c \in \Acomm} 
\picomm(o_{t-1})(c) \mu_{M_\piact}(\zeta) \cdot \log \left( \mu^{c}(a_t^c s_t^c \mid s_0 a_1 s_1 \ldots a_{t-1} s_{t-1}) \right) \nonumber\\ 
- & \sum\limits_{\zeta} \sum\limits_{c \in \Acomm} 
\picomm(o_{t-1})(c) \mu_{M_\piact}(\zeta) \cdot \log \left( \prod_{i \notin c} \mu^{i}(a_t^i s_t^i \mid s_0 a_1 s_1 \ldots a_{t-1} s_{t-1}) \right).
\end{align}
\end{subequations}
\clearpage

By taking the sum over \(t\), we have the following:

\begin{subequations} \label{eq:long_equation5}
\begin{align}
 D_{KL}\left(\Gamma_{M_\piact} \parallel \Gamma_{\widehat{M}_{\pi}}\right) \leq & \sum\limits_{t=1}^{\infty} \sum\limits_{\zeta} 
\mu_{M_\piact}(\zeta) \cdot
\log \left( \mu_{M_\piact}\left(a_t s_t \mid s_0 a_1 s_1 \ldots a_{t-1} s_{t-1} \right) \right) \nonumber\\
& - \sum\limits_{t=1}^{\infty} \sum\limits_{\zeta} 
\sum\limits_{c \in \Acomm} \picomm(o_{t-1})(c) 
\mu_{M_\piact}(\zeta) \cdot \log \left( \mu^{c}\left(a_t^c s_t^c \mid s_0 a_1 s_1 \ldots a_{t-1} s_{t-1} \right) \right) \nonumber\\
& - \sum\limits_{t=1}^{\infty} \sum\limits_{\zeta} 
\sum\limits_{c \in \Acomm} \picomm(o_{t-1})(c) 
\mu_{M_\piact}(\zeta) \cdot \log \left( \prod_{i \notin c} 
\mu^{i}\left(a_t^i s_t^i \mid s_0 a_1 s_1 \ldots a_{t-1} s_{t-1} \right) \right) \\
= & \sum\limits_{t=1}^{\infty} \sum\limits_{\zeta} 
\mu_{M_\piact}(\zeta) \cdot \log \left( \mu_{M_\piact}\left(a_t s_t \mid s_0 a_1 s_1 \ldots a_{t-1} s_{t-1} \right) \right) \nonumber\\
& - \sum\limits_{t=1}^{\infty} \sum\limits_{c \in \Acomm} \Bigg( 
\sum\limits_{\zeta} 
\picomm(o_{t-1})(c) \mu_{M_\piact}(\zeta) \cdot \log \left( \mu^{c}\left(a_t^c s_t^c \mid s_0 a_1 s_1 \ldots a_{t-1} s_{t-1} \right) \right) \nonumber\\
& \quad \hspace{2.5cm} + \sum\limits_{\zeta} 
\picomm(o_{t-1})(c) \mu_{M_\piact}(\zeta) \cdot \log \left( \prod_{i \notin c} 
\mu^{i}\left(a_t^i s_t^i \mid s_0 a_1 s_1 \ldots a_{t-1} s_{t-1} \right) \right) 
\Bigg) \\
= & \sum\limits_{t=1}^{\infty} \sum\limits_{\zeta} 
\mu_{M_\piact}(\zeta) 
\cdot \log \left( \mu_{M_\piact}\left(a_t s_t \mid s_0 a_1 s_1 \ldots a_{t-1} s_{t-1} \right) \right) \nonumber\\
& - \sum\limits_{t=1}^{\infty} \sum\limits_{c \in \Acomm} \Bigg( 
\sum\limits_{\zeta} 
\picomm(o_{t-1})(c) \mu_{M_\piact}(\zeta) \cdot \log \left( \mu^{c}\left(a_t^c s_t^c \mid s_{t-1} \right) \right) \nonumber\\
& \quad\hspace{2.5cm} + \sum\limits_{\zeta} 
\picomm(o_{t-1})(c) \mu_{M_\piact}(\zeta) \cdot \log \left( \prod_{i \notin c} 
\mu^{i}\left(a_t^i s_t^i \mid s_{t-1} \right) \right) 
\Bigg).  \label{eq:*}
\end{align}
\end{subequations}
\clearpage

Applying the definition of $G^c$ and $G^i$ we obtain
\begin{subequations} \label{eq:long_equation6}
\begin{align}
& \sum_{c \in \Acomm} G^c(\pi_{\text{comm}},\pi_{\text{act}}) 
+ \sum_{i \in [N]} G^i(\pi_{\text{comm}},\pi_{\text{act}})\nonumber\\
& =\sum\limits_{c \in \Acomm} \sum\limits_{t=1}^{\infty} 
\sum\limits_{o_{t-1}, l_{t-1}^c} \prob(o_{t-1}, l_{t-1}^c) 
\cdot \picomm(o_{t-1})(c) \cdot L(c, o_{t-1}, l_{t-1}^c, t) \nonumber\\
&\quad + \sum_{i \in [N]} \sum\limits_{t=1}^{\infty} 
\sum\limits_{o_{t-1}, l_{t-1}^i} \prob(o_{t-1}, l_{t-1}^i) \cdot \left( \sum\limits_{c \in \Acomm, i\not\in c} \picomm(o_{t-1})(c) \right) \cdot L(i, o_{t-1}, l_{t-1}^i, t) \\
&= \sum\limits_{t=1}^{\infty} \sum\limits_{c \in \Acomm} \Bigg( 
\sum\limits_{o_{t-1}, l_{t-1}^c} \prob(o_{t-1}, l_{t-1}^c) 
\cdot \picomm(o_{t-1})(c) \cdot L(c, o_{t-1}, l_{t-1}^c, t) \nonumber\\
&\quad \hspace{2.5cm} + \sum_{i \notin c} \sum\limits_{o_{t-1}, l_{t-1}^i} 
\prob(o_{t-1}, l_{t-1}^i) \cdot \picomm(o_{t-1})(c) \cdot L(i, o_{t-1}, l_{t-1}^i, t) \Bigg)\\
&= \sum\limits_{t=1}^{\infty} \sum\limits_{c \in \Acomm} \Bigg( 
\sum\limits_{o_{t-1}, l_{t-1}^c} \sum\limits_{a_t^c} \sum\limits_{o_t^c, l_t^c} \picomm(o_{t-1})(c) \cdot \prob(o_{t-1}, l_{t-1}^c) \cdot  \prob(a_t^c,o_t^c,l_t^c \mid o_{t-1},l_{t-1}^c) \cdot  \log(\prob(a_t^c,o_t^c,l_t^c \mid o_{t-1},l_{t-1}^c)) \nonumber\\
&\quad \hspace{2.5cm} + \sum_{i \notin c} \sum\limits_{o_{t-1}, l_{t-1}^i} \sum\limits_{a_t^i} \sum\limits_{o_t^i, l_t^i} \picomm(o_{t-1})(c) \cdot \prob(o_{t-1}, l_{t-1}^i) \cdot \prob(a_t^i,o_t^i,l_t^i \mid o_{t-1},l_{t-1}^i)\cdot \log(\prob(a_t^i,o_t^i,l_t^i \mid o_{t-1},l_{t-1}^i)) \Bigg)\\
&= \sum\limits_{t=1}^{\infty} \sum\limits_{c \in \Acomm} \Bigg( 
\sum\limits_{o_{t-1}, l_{t-1}^c} \sum\limits_{a_t^c} \sum\limits_{o_t^c, l_t^c} \picomm(o_{t-1})(c) \cdot \prob(o_{t-1}, l_{t-1}^c, a_t^c, o_t^c,l_t^c)\cdot \log(\prob(a_t^c,o_t^c,l_t^c \mid o_{t-1},l_{t-1}^c)) \nonumber\\
&\quad \hspace{2.5cm} + \sum_{i \notin c} \sum\limits_{o_{t-1}, l_{t-1}^i} \sum\limits_{a_t^i} \sum\limits_{o_t^i, l_t^i} \picomm(o_{t-1})(c) \cdot \prob(o_{t-1}, l_{t-1}^i, a_t^i,o_t^i,l_t^i) \cdot \log(\prob(a_t^i,o_t^i,l_t^i \mid o_{t-1},l_{t-1}^i)) \Bigg).
\end{align}
\end{subequations}
To reduce notational complexity, we use \(\acts^{\bar{i}}= \acts^{1} \times \ldots \times \acts^{i-1} \times \acts^{i+1} \times \ldots \times \acts^{n}\) to represent the joint actions of agent \(i\)'s teammates, excluding agent \(i\) itself. Similarly, we denote the publicly observable states and local states of agent \(i\)'s teammates, excluding agent \(i\) itself, as \(\ostates^{\bar{i}}\) and \(\lstates^{\bar{i}}\), respectively.
In a similar manner, for a group of agents \(c\), we denote the actions, publicly observable states, and local states of the teammates of group \(c\), excluding the agents within \(c\) itself, as \(\acts^{\bar{c}}\), \(\ostates^{\bar{c}}\), and \(\lstates^{\bar{c}}\), respectively.
By applying the definition of marginal probability,
\clearpage

\begin{subequations} \label{eq:long_equation7}
\begin{align}
& \sum_{c \in \Acomm} G^c(\pi_{\text{comm}},\pi_{\text{act}}) 
+ \sum_{i \in [N]} G^i(\pi_{\text{comm}},\pi_{\text{act}}) \nonumber\\
&= \sum\limits_{t=1}^{\infty} \sum\limits_{c \in \Acomm} \Bigg( 
\sum\limits_{o_{t-1}, l_{t-1}^c} \sum\limits_{a_t^c} 
\sum\limits_{o_t^c, l_t^c} \sum\limits_{\ldots} 
\picomm(o_{t-1})(c) \nonumber\\
& \quad \hspace{3.0cm} \cdot \prob(o_0, l_0, a_1, o_1, l_1, \ldots, 
a_{t-1}, o_{t-1}, l_{t-1}^{\bar{c}}, l_{t-1}^c, a_t^{\bar{c}}, a_t^c, o_t^{\bar{c}}, o_t^c, l_t^{\bar{c}}, l_t^c, a_{t+1}, o_{t+1}, l_{t+1}, \ldots) \nonumber\\
& \quad \hspace{3.5cm} \cdot \log(\prob(a_t^c, o_t^c, l_t^c \mid o_{t-1}, l_{t-1}^c)) \nonumber\\
& \hspace{2.5cm} + \sum_{i \notin c} \sum\limits_{o_{t-1}, l_{t-1}^i} 
\sum\limits_{a_t^i} \sum\limits_{o_t^i, l_t^i} \sum\limits_{\ldots} 
\picomm(o_{t-1})(c) \nonumber\\
& \quad \hspace{3.0cm} \cdot \prob(o_0, l_0, a_1, o_1, l_1, \ldots, 
a_{t-1}, o_{t-1}, l_{t-1}^{\bar{i}}, l_{t-1}^i, a_t^{\bar{i}}, a_t^i, o_t^{\bar{i}}, o_t^i, l_t^{\bar{i}}, l_t^i, a_{t+1}, o_{t+1}, l_{t+1}, \ldots) \nonumber\\
& \quad  \hspace{3.5cm} \cdot \log(\prob(a_t^i, o_t^i, l_t^i \mid o_{t-1}, l_{t-1}^i)) 
\Bigg)\\
&= \sum\limits_{t=1}^{\infty} \sum\limits_{c \in \Acomm} \Bigg( 
\sum\limits_{o_0 l_0 a_1 o_1 l_1 a_2\ldots } 
\picomm(o_{t-1})(c) \cdot \prob(o_0 l_0, a_1, o_1 l_1, a_2,\ldots) \cdot 
\log(\prob(a_t^c, o_t^c, l_t^c \mid o_{t-1}, l_{t-1}^c)) \nonumber\\
& \hspace{2.5cm} + \sum_{i \notin c} \sum\limits_{o_0 l_0 a_1 o_1 l_1 a_2\ldots} \picomm(o_{t-1})(c) \cdot \prob(o_0 l_0, a_1, o_1 l_1, a_2,\ldots) \cdot 
\log(\prob(a_t^i, o_t^i, l_t^i \mid o_{t-1}, l_{t-1}^i)) 
\Bigg)\\
&= \sum\limits_{t=1}^{\infty} \sum\limits_{c \in \Acomm} \Bigg( 
\sum\limits_{\zeta} 
\picomm(o_{t-1})(c) \cdot \prob(\zeta) \cdot \log(\prob(a_t^c, o_t^c, l_t^c \mid o_{t-1}, l_{t-1}^c)) \nonumber\\
& \hspace{2.5cm} + \sum_{i \notin c} \sum\limits_{\zeta} \picomm(o_{t-1})(c) \cdot \prob(\zeta) \cdot \log(\prob(a_t^i, o_t^i, l_t^i \mid o_{t-1}, l_{t-1}^i)) 
\Bigg)\\
&= \sum\limits_{t=1}^{\infty} \sum\limits_{c \in \Acomm} \Bigg( 
\sum\limits_{\zeta} 
\picomm(o_{t-1})(c) \cdot \prob(\zeta) \cdot \log(\prob(a_t^c, o_t^c, l_t^c \mid o_{t-1}, l_{t-1}^c)) \nonumber\\
& \hspace{2.5cm} + \sum\limits_{\zeta} \picomm(o_{t-1})(c) \cdot \prob(\zeta) \cdot \log( \prod\limits_{i \notin c} \prob(a_t^i, o_t^i, l_t^i \mid o_{t-1}, l_{t-1}^i)) 
\Bigg). \label{eq:**}
\end{align}
\end{subequations}

A comparison between (\ref{eq:*}) and (\ref{eq:**}) reveals that

\begin{align*}
& D_{KL}\left(\Gamma_{M_\piact} \parallel \Gamma_{\widehat{M}_{\pi}}\right) \leq \sum_{c \in \Acomm} G^c(\pi_{\text{comm}},\pi_{\text{act}}) 
+ \sum_{i \in [N]} G^i(\pi_{\text{comm}},\pi_{\text{act}}) - H(X).
\end{align*}

Hence, we have
\begin{align*}
\prob_{M_{\piact}}((\neg\savoid) \mathcal{U} \starget) - \prob_{\widehat M_{\pi}}((\neg\wsavoid) \mathcal{U} \wstarget) \nonumber & \leq \sqrt{1-\exp\left(-D_{KL}\left(\Gamma_{M_\piact} \parallel \Gamma_{\widehat{M}_{\pi}}\right) \right)}\\
& \leq \sqrt{1-\exp\left(-D_{(\pi_{\text{comm}},\pi_{\text{act}})}  \right)}.     
\end{align*}
\end{proof}
\clearpage

\expressionentropy*
\begin{proof}
Using the chain rule for conditional entropy, we compute the entropy at time \(t \geq 1\) as \eqref{eq:step_entropy_joint}.
\begin{equation}
\begin{aligned}
& H \left(A_t S_t \mid S_0A_1S_1\ldots A_{t-1}S_{t-1}\right)\\
& = -\sum_{a',s'} \sum_{s_0}\sum_{a_1,s_1} \ldots \sum_{a_{t-1} s_{t-1}} 
\Big( \prob( A_t=a',S_t=s',S_0=s_0,\ldots, A_{t-1}=a_{t-1}, S_{t-1}=s_{t-1}) \\
& \hspace{4.5cm} \cdot \log \prob \left( A_t=a',S_t=s'\mid S_0=s_0,\ldots, A_{t-1}=a_{t-1}, S_{t-1}=s_{t-1}\right) \Big)\\
& = -\sum_{a',s'} \sum_{s} \prob \left( A_t=a',S_t=s'\mid S_{t-1} =s\right) \cdot \prob \left(S_{t-1} =s\right) \cdot
\log \prob \left( A_t=a',S_t=s' \mid S_{t-1}=s\right) \\
& = \sum_{s}   \prob \left(S_{t-1} =s\right) \cdot L_M(s)\\
\end{aligned}\label{eq:step_entropy_joint}
\end{equation}

where we define
\begin{align*}
        L_M(s)=&-\sum_{a',s'}\piact(s)(a')P(s,a')(s') \cdot \log (\piact(s)(a')P(s,a')(s')).
\end{align*}

Applying \eqref{eq:step_entropy_joint} to the chain rule for joint entropy over an infinite time we obtain 

\begin{align*}
& H(S_0) + \sum_{t=1}^{\infty} H \left(A_t  S_t  \middle| S_0 A_1 S_1 \dots A_{t-1} S_{t-1}  \right) \\
& = 0 + \sum_{t=1}^{\infty}  \sum_{s}   \prob\left(S_{t-1} =s\right) \cdot L_M(s)\\
& =  \sum_{s}    L_M(s) \sum_{t=1}^{\infty}  \prob \left(S_{t-1} =s\right)\\ 
& = \sum_{s}    L_M(s) \cdot \nu_s\\
& =  -\sum_{s}  \nu_s \sum_{a',s'}\piact(s)(a') \cdot P(s,a')(s') \cdot \log (\piact(s)(a')\cdot P(s,a')(s')) \\
& = -\sum_{s,a',s'}  (\nu_s \cdot \piact(s)(a')) \cdot P(s,a')(s') \cdot \log (\piact(s)(a')\cdot P(s,a')(s')) \\
& = - \sum_{s,a',s'}  \nu_{s,a'}  \cdot P(s,a')(s') \cdot \log (\piact(s)(a')\cdot P(s,a')(s')) \\
& = - \sum_{s,a',s'}  \nu_{s,a'}  \cdot P(s,a')(s') \cdot \left(\log \left( \frac{\nu_{s,a'}}{\sum_{b}{\nu_{s,b}}} \right) + \log P(s,a')(s')\right)\\
& = - \left(\sum_{s,a',s'}  \nu_{s,a'}   \cdot P(s,a')(s')\cdot \log \left( \frac{\nu_{s,a'}}{\sum_{b}{\nu_{s,b}}} \right)\right) - \left(\sum_{s,a',s'}  \nu_{s,a'}  \cdot P(s,a')(s') \cdot  \log P(s,a')(s'))\right)\\
& = - \left(\sum_{s,a'}  \nu_{s,a'}  \cdot \log \left( \frac{\nu_{s,a'}}{\sum_{b}{\nu_{s,b}}} \right)\right) - \left(\sum_{s,a',s'}  \nu_{s,a'}  \cdot P(s,a')(s') \cdot  \log P(s,a')(s'))\right)\\
\end{align*}

Note that, with our assumption of a single initial state in each MMDP, it is consistently true that \(H(S_0) = 0\). 
\end{proof}

\expressionsgs*

\begin{proof}
We prove the claim for $i \in [N]$, the proof for $c \subseteq [N]$ is analogous.

By definition, 
$
G^i(\pi_{\text{comm}},\pi_{\text{act}})  = 
\sum_{t=1}^{\infty} \sum\limits_{o \in \ostates, l^i\in\lstates^i }
\prob(O_{t-1} = o,L_{t-1}^i = l^i)\cdot w(o,i) \cdot L(i,o,l^i,t),
$
where  $ L(i,o,l^i,t)$ is defined by equality \eqref{eq:long_equation} below.
\begin{equation}
\label{eq:long_equation}
\begin{array}{llll}
L(i,o,l^i,t) & = -\sum\limits_{a^i \in \acts^i, o_1^i \in \ostates^i,l_1^i \in \lstates^i} & \prob(A_t^i =a^i,O^i_t=o^i_1,L_t^i=l_1^i \mid O_{t-1}=o, L_{t-1}^i=l^i )\cdot\\&&
\log\big(\prob(A_t^i =a^i,O_t^i=o_1^i,L_t^i=l_1^i \mid O_{t-1}=o,L_{t-1}^i=l^i )\big).
\end{array}
\end{equation}
From \eqref{eq:long_equation} we obtain \eqref{eq:long_equation2} below.
\begin{equation}
\label{eq:long_equation2}
\begin{array}{llll}
L(i,o,l^i,t)\!\!\! & = 
-\hspace{-0.5cm}\sum\limits_{a^i \in \acts^i, o_1^i \in \ostates^i,l_1^i \in \lstates^i}\hspace{-0.2cm} & 
\prob(A_t^i =a^i  | O_{t-1}=o, L_{t-1}^i=l^i ) \!\cdot\!  
\prob(O^i_t=o^i_1,L_t^i=l_1^i|  O_{t-1}=o, L_{t-1}^i=l^i, A_t^i =a^i ) \cdot\\&&
\log\big(
\prob(A_t^i =a^i  \mid O_{t-1}=o, L_{t-1}^i=l^i ) \cdot  
\prob(O^i_t=o^i_1,L_t^i=l_1^i\mid  O_{t-1}=o, L_{t-1}^i=l^i, A_t^i =a^i )
\big)\\&=
-\sum\limits_{a^i \in \acts^i, o_1^i \in \ostates^i,l_1^i \in \lstates^i} & 
\prob(A_t^i =a^i  \mid O_{t-1}=o, L_{t-1}^i=l^i ) \cdot  
\prob(O^i_t=o^i_1,L_t^i=l_1^i\mid  O^i_{t-1}=o^i, L_{t-1}^i=l^i, A_t^i =a^i ) \cdot\\&&
\log\big(
\prob(A_t^i =a^i  \mid O_{t-1}=o, L_{t-1}^i=l^i ) \cdot  
\prob(O^i_t=o^i_1,L_t^i=l_1^i\mid  O^i_{t-1}=o^i, L_{t-1}^i=l^i, A_t^i =a^i )
\big)\\&
\leq 
-\sum\limits_{a^i \in \acts^i, o_1^i \in \ostates^i,l_1^i \in \lstates^i} & 
\frac{\nu_{o,l^i,a^i}}{\sum\limits_{b^i \in \acts^i} \nu_{o,l^i,b^i}} \cdot  
P^i(o^i,l^i,A^i)(o_1^i,l_1^i) \cdot
\log\big(
\frac{\nu_{o,l^i,a^i}}{\sum\limits_{b^i \in \acts^i} \nu_{o,l^i,b^i}} \cdot  
P^i(o^i,l^i,A^i)(o_1^i,l_1^i)
\big)
\end{array}
\end{equation}

Substituting \eqref{eq:long_equation2} in the definition of $G^i$, we obtain \eqref{eq:long_equation3} below.

\begin{equation}
\label{eq:long_equation3}
\begin{array}{llrl}
G^i(\pi_{\text{comm}},\pi_{\text{act}})   & \leq &
-\sum_{t=1}^{\infty} \sum\limits_{o \in \ostates, l^i\in\lstates^i,a^i \in \acts^i, o_1^i \in \ostates^i,l_1^i \in \lstates^i }
& \prob(O_{t-1} = o,L_{t-1}^i = l^i)\cdot w(o,i) \cdot 
\frac{\nu_{o,l^i,a^i}}{\sum\limits_{b^i \in \acts^i} \nu_{o,l^i,b^i}} \cdot  \\&&&
P^i(o^i,l^i,A^i)(o_1^i,l_1^i) \cdot\\&&&
\log\big(
\frac{\nu_{o,l^i,a^i}}{\sum\limits_{b^i \in \acts^i} \nu_{o,l^i,b^i}} \cdot  
P^i(o^i,l^i,A^i)(o_1^i,l_1^i)
\big)\\&\leq&
- \sum\limits_{o \in \ostates, l^i\in\lstates^i,a^i \in \acts^i, o_1^i \in \ostates^i,l_1^i \in \lstates^i }
& \nu_{o,l^i,a^i}\cdot w(o,i) \cdot  
P^i(o^i,l^i,A^i)(o_1^i,l_1^i) \cdot\\&&&
\log\big(
\frac{\nu_{o,l^i,a^i}}{\sum\limits_{b^i \in \acts^i} \nu_{o,l^i,b^i}} \cdot  
P^i(o^i,l^i,A^i)(o_1^i,l_1^i)
\big)\\&\leq&
- \sum\limits_{o \in \ostates, l^i\in\lstates^i,a^i \in \acts^i }
& \nu_{o,l^i,a^i}\cdot w(o,i) \cdot  
\log\big(
\frac{\nu_{o,l^i,a^i}}{\sum\limits_{b^i \in \acts^i} \nu_{o,l^i,b^i}} 
\big)\\&&
- \sum\limits_{o \in \ostates, l^i\in\lstates^i,a^i \in \acts^i, o_1^i \in \ostates^i,l_1^i \in \lstates^i }
& \nu_{o,l^i,a^i}\cdot w(o,i) \cdot  
P^i(o^i,l^i,A^i)(o_1^i,l_1^i) \cdot
\log\big(
P^i(o^i,l^i,A^i)(o_1^i,l_1^i)\big)
\end{array}
\end{equation}

Finally, note that 
$
w(o,i) = \sum\limits_{c \in \Acomm, i\not\in c}\picomm(o)(c)= \sum\limits_{c \in \Acomm, i\not\in c}\frac{\nu_{o,c}}{\sum\limits_{c' \in \Acomm}\nu_{o,c'}}.
$
\end{proof}

\clearpage	

\section{Details on Policy Computation}\label{sec:app-optimization}
Here we give the details of the two steps of our approach for computing positional action and communication policies.

\subsubsection{Optimistic Optimal Value for Reach-Avoid Probability}
In the first step, we use a standard method to compute the optimal value $v^*(M,(\neg\savoid) \mathcal{U} \starget)$ for the reach-avoid probability assuming unrestricted communication. 
The problem at this stage is formulated as a linear program with occupancy measures \(x_{s,a}\) as the variables, and the objective is to maximize the reach-avoid probability. 
We solve the following optimization problem to determine the optimal reach-avoid probability value under a centralized policy execution. 
Subsequently, we employ this optimal value in a constraint in the optimization problem solved at the second stage.

\begin{align*}
& v^* = \max_{x_{s,a}}{\sum_{s \in \states\setminus(\states_{avoid}\cup \states_{target} )} \sum_{a \in \acts } \sum_{s' \in \states_{target}}{x_{s,a}P(s,a)(s')}}\\
& \sum_{a \in \acts} {{x_{s,a}} = \sum_{\substack{s'\in \states \\ b \in \acts }} {x_{s',b}P(s',b)(s)}+{\mathbb{1}}_{\{s_\init=s\}}} \; \forall s \in \states \setminus (\states_{avoid}\cup \states_{target} )\\
& {x_{s,a}}  \geq 0 \; \forall s \in \states \setminus (\states_{avoid}\cup \states_{target} ), a \in \acts\\
& {x_{s,a}}  = 0 \; \forall  s\in  (\states_{avoid}\cup \states_{target} ), a \in \acts\\
\end{align*}

\clearpage	

\paragraph{Cost Minimization}

In the second step,  the decision variables are occupancy measures \((x_{o,l,a}, x_{o,c})\), aiming to optimize the additional cost associated with communication while determining a pair of policies for both communication and action.

\begin{align*}
   & \min_{(x_{o,l,a}, x_{o,c})}{\bar{d} = \sum_{i \in [n]} g^i + \sum_{c \in \Acomm} g^c- h}\\
   & g^i = - \left(\sum_{o,l^i,a^i}  x_{o,l^i,a^i}\cdot  w(o,i) \cdot \log \left( \frac{x_{o,l^i,a^i}}{\sum_{b^i}{x_{o,l^i,b^i}}} \right)\right) - \Bigg(\sum_{\substack{
   o,l^i,a^i \\ o_1^i,l^i_1}}  x_{o,l^i,a^i} \cdot  w(o,i)  \cdot P^i(o^i,l^i,a^i)(o_1^i,l^i_1) \cdot  \log P^i(o^i,l^i,a^i)(o_1^i,l^i_1)\Bigg)\\
   &  g^c = - \left(\sum_{o,l^c,a^c}  x_{o,l^c,a^c}  \cdot  w(o,c) \cdot \log \left( \frac{x_{o,l^c,a^c}}{\sum_{b^c}{x_{o,l^c,b^c}}} \right)\right) -\Bigg(\sum_{\substack{
   o,l^c,a^c\\o_1^c,l^c_1}}  x_{o,l^c,a^c} \cdot  w(o,c) \cdot P^c(o^c,l^c,a^c)(o_1^c,l^c_1) \cdot  \log P^c(o^c,l^c,a^c)(o_1^c,l^c_1))\Bigg)\\
   & h = - \left(\sum_{s,a'}  x_{s,a'}  \cdot \log \left( \frac{x_{s,a'}}{\sum_{b}{x_{s,b}}} \right)\right) - \left(\sum_{s,a',s'}  x_{s,a'}  \cdot P(s,a')(s') \cdot  \log P(s,a')(s'))\right)\\
   & w(o,i) =  \sum\limits_{c \in \Acomm, i\not\in c}\frac{x_{o,c}}{\sum\limits_{c' \in \Acomm}x_{o,c'}} \; \forall  o \in \ostates,  i \in [N]\\
   & w(o,c) =  \frac{x_{o,c}}{\sum\limits_{c' \in \Acomm}x_{o,c'}} \; \forall o \in \ostates, c \in \Acomm\\
   & v^* \leq \sum_{\left(o,l\right) \in \states\setminus(\states_{avoid}\cup \states_{target} )} \sum_{a \in \acts } \sum_{\left(o',l'\right) \in \states_{target}}{x_{o,l,a}P(o,l,a)(o',l')}\\
   & \sum_{a \in \acts \cup \{a_\alpha\}}{{x_{o,l,a}} = \sum_{\substack{(o',l')\in \states \\ b \in \acts \cup \{a_\alpha\} }} {x_{o',l',b}P(o',l',b)(o,l)}+{\mathbb{1}}_{\{s_\init=s\}}} \; \forall (o,l) \in \states \\
   & x_{\left( o,l \right), a}  \geq 0 \; \forall \left(o, l\right) \in \states, a \in \acts\cup\{a_\alpha\}\\
   & x_{(o_\alpha,l_\alpha),a}  = 0 \; \forall a \in \acts\\
   & x_{o, c}  \geq 0 \; \forall o \in \ostates, c \in \Acomm\\
   & x_{o_\alpha,c}  = 0 \; \forall c\in \Acomm\\
   & \sum_{l \in \lstates, a \in \acts}{x_{o,l,a}} = \sum_{c \in \Acomm}x_{o,c} \; \forall o \in \ostates\\
\end{align*}

\newpage 

\clearpage	 

\section{Detailed Description of Benchmarks}\label{sec:app-benchmarks}
\subsubsection{Scenario \#1 with Navigation Tasks}
Consider the environment in Figure~\ref{fig:Env_1_figure}, which is a 4 × 3 grid. The three robots \(R1\), \(R2\), and \(R3\) are initialized as marked in the figure, and their tasks are to navigate to their target locations, \(T1\), \(T2\), and \(T3\), respectively.
Each of \(R1\) and \(R2\) has two potential target locations.  Once each of the robots has reached one of their target locations,  the team's task is complete. At any given time step, only two out of the three robots can communicate and share precise locations and local states. They make decisions on communication by sharing public information, including their respective regions within the environment, which is partitioned into three regions labeled \(o=0\), \(o=1\), and \(o=2\).
Following the communication action, each robot selects one of five distinct actions: move North, move East, move South, move West, or remain in the current cell. 
 If the robot selects an action to move (North, East, South, or West), it proceeds to the desired next state with the probability of \(0.9\) and fails to move in the selected state with the probability of \(0.1\). If the robot fails to move to the desired state, it remains in the current state.  If the robot selects the remaining action, it stays in the current state with probability \(1\). If the selected action results in an invalid move (e.g., hitting a wall), all probability is assigned to staying in the current state.

The method proposed in Section~\ref{sec:synthesis} generates a pair of policies where the action policy is optimal under full communication while creating a robust system under communication restrictions. 
Figures~\ref{fig:Env_1_Heat_map_full_comm}--\ref{fig:Env_1_Heat_map_limited_comm} present the heat maps of the occupancy measures for the robots  for the joint action policy synthesized 
without and with minimizing communication respectively.

This example shows that our approach produces action and communication policies that achieve zero communication costs while maintaining optimal reach-avoid probabilities. The reach-avoid probability under full communication is \(0.99\), which can be achieved under restricted communication by our approach with zero communication cost. In this scenario, the generated policies suggest a communication policy with a probability of \(1\) between robots \(R1\) and \(R2\), which differs from the one based on full communication. Therefore, the communication policy effectively identifies the robots that need to communicate.

\begin{figure}
    \centering
    \begin{subfigure}[b]{0.32\textwidth}
        \includegraphics[width=\linewidth]{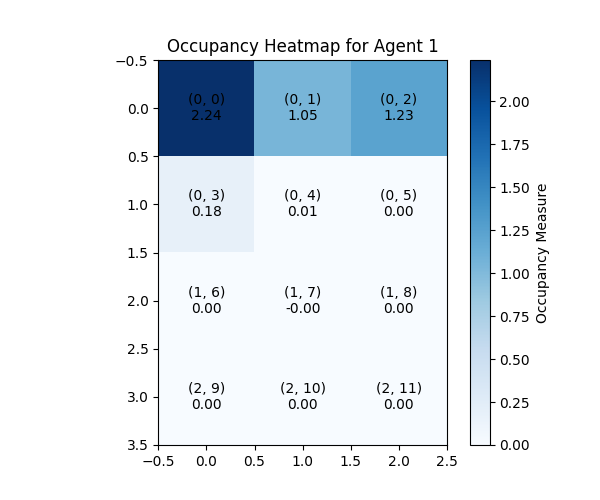}
        \caption{Occupancy measures for agent $R1$}
    \end{subfigure}
    \hfill
    \begin{subfigure}[b]{0.32\textwidth}
        \includegraphics[width=\linewidth]{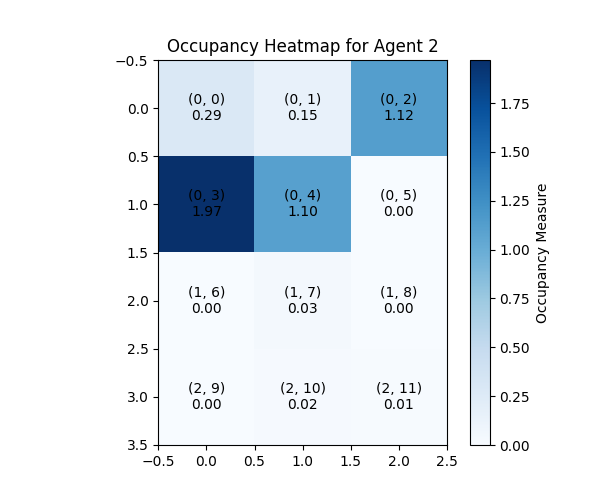}
        \caption{Occupancy measures for agent $R2$}
    \end{subfigure}
    \hfill
    \begin{subfigure}[b]{0.32\textwidth}
        \includegraphics[width=\linewidth]{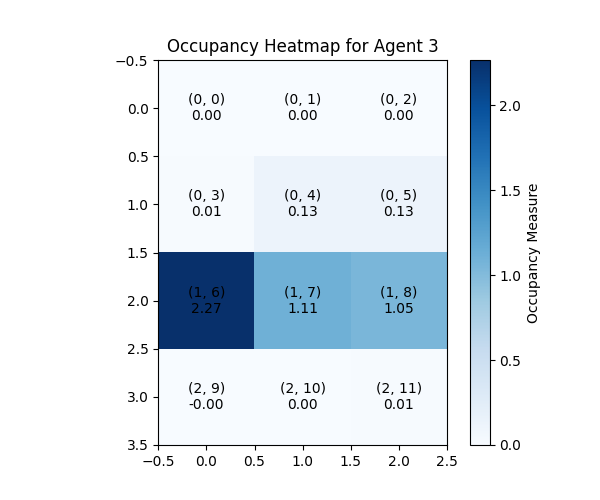}
        \caption{Occupancy measures for agent $R3$}
    \end{subfigure}
    \caption{\textbf{Scenario \#1.} Heat maps of the occupancy measures for the policy computed without minimizing communication.}
    \label{fig:Env_1_Heat_map_full_comm}
\end{figure}

\begin{figure}
    \centering
    \begin{subfigure}[b]{0.32\textwidth}
        \centering
        \includegraphics[width=\linewidth]{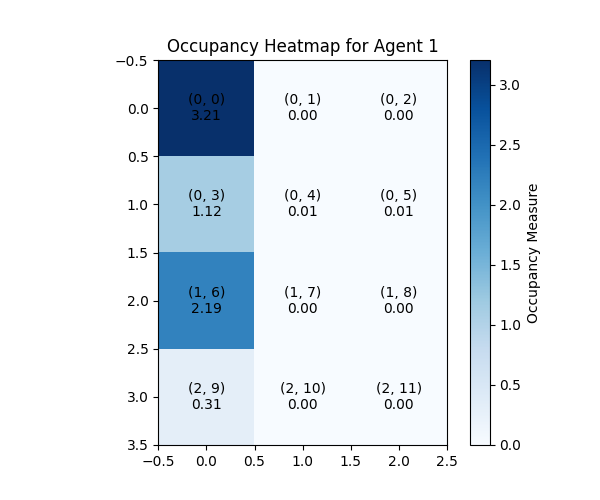}
        \caption{Occupancy measures for agent $R1$}
    \end{subfigure}
    \hfill
    \begin{subfigure}[b]{0.32\textwidth}
        \centering
        \includegraphics[width=\linewidth]{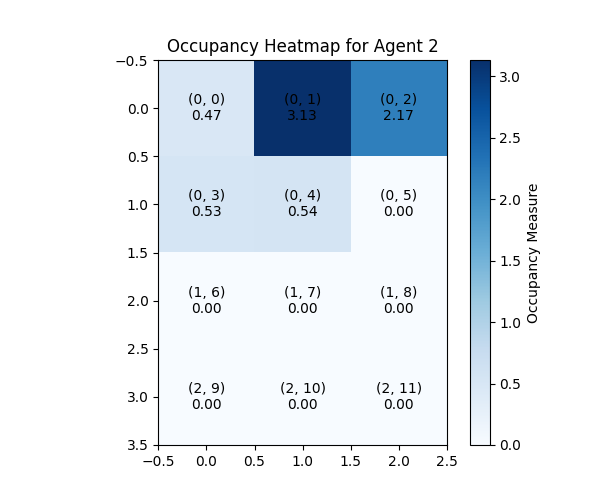}
        \caption{Occupancy measures for agent $R2$}
    \end{subfigure}
    \hfill
    \begin{subfigure}[b]{0.32\textwidth}
        \centering
        \includegraphics[width=\linewidth]{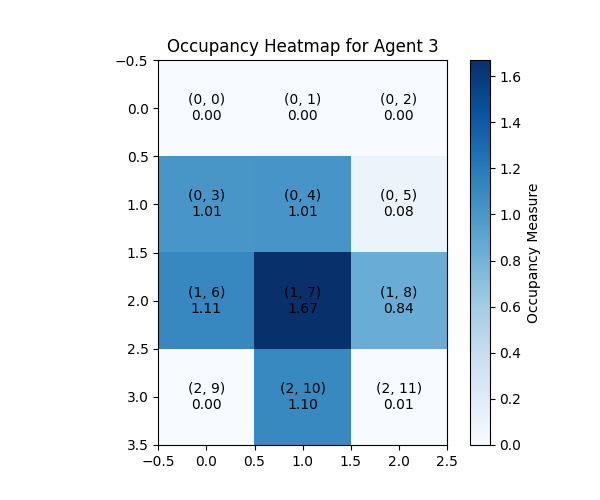}
        \caption{Occupancy measures for agent $R3$}
    \end{subfigure}
    \caption{\textbf{Scenario \#1.} Heat maps of the occupancy measures for the policy computed when minimizing communication.}
    \label{fig:Env_1_Heat_map_limited_comm}
\end{figure}
\clearpage

\subsubsection{Scenario \#2 with a Swarm Intersection}
Consider the environment depicted in Figure~\ref{fig:Environment_2}, where three robots \(R1\), \(R2\), and \(R3\) must navigate to their respective target locations, labeled \(T1\), \(T2\), and \(T3\). The environment comprises 12 cells, labeled from \(0\) to \(11\), as illustrated in Figure~\ref{fig:Env_2_local_states}, and is divided into three regions, denoted by \(o=0\), \(o=1\), and \(o=2\), as shown in Figure~\ref{fig:Environment_2_region}. Each robot completes its task upon reaching one of its designated target locations. The set of all possible joint targets is presented in Table~\ref{tab:Env_2_targets}. The objective is to reach these targets while avoiding collisions with the highest possible probability.
Table~\ref{tab:Env_2_agent_transitions} provides the transition probabilities that describe how robots move through the environment. In this scenario, a congested intersection introduces a high risk of collision, making inter-agent communication essential for coordinating movement and ensuring safe navigation.
We compare the action policy computed by our approach against the approach based on minimizing the total correlation as the objective function. At any given time step, only two out of the three robots are permitted to communicate and exchange precise locations. 
The heat maps of the occupancy measures for the robots, under the joint policies synthesized by our method and by minimizing the total correlation, are shown in Figures~\ref{fig:Env_2_Heat_map_limited_comm}--\ref{fig:Env_2_Heat_map_total_correlation}.
Both the action policy synthesized by total correlation and our method achieve a reach-avoid probability of $1$. However, our approach ensures zero communication cost by selecting an appropriate set of communicating robots. In contrast, the policy derived from total correlation violates the communication restriction at time $t=2$, as it requires coordination among three agents at that time. In this scenario, the minimum total correlation is $0.591$, while the total correlation under our method is $0.693$. This demonstrates that, although an action policy with minimal dependency among agents may exist, minimizing the total correlation alone may fail to find a valid policy that adheres to communication constraints, leading to additional communication costs.
\begin{figure}
    \centering
    \includegraphics[scale=0.075]{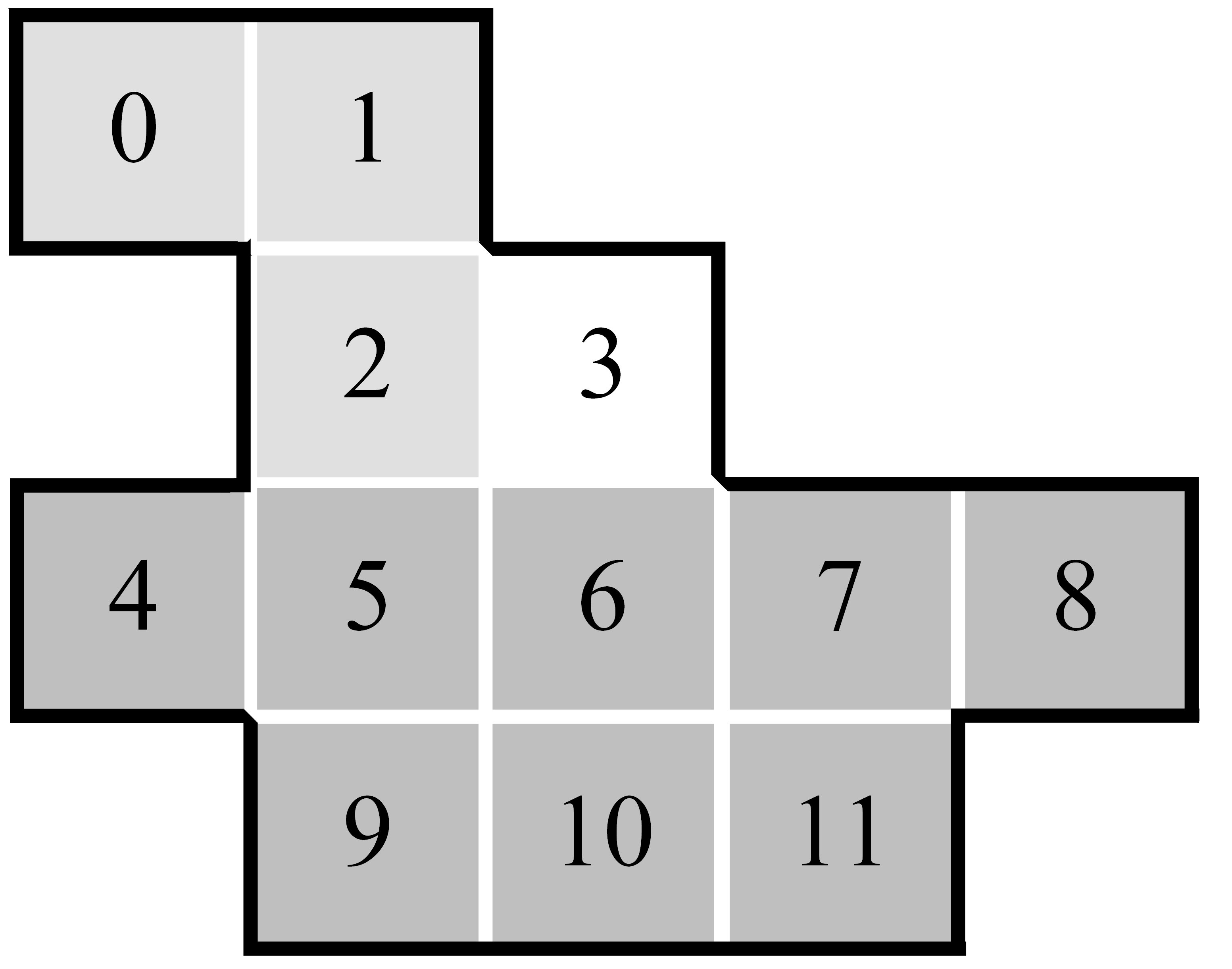}
    \vspace{-.2cm}
    \caption{Local states labels used in \textbf{Scenario \#2}.}
    \label{fig:Env_2_local_states}
\end{figure}
\begin{table}[H]
\centering
\begin{tabular}{cccccccc}
\hline
\multicolumn{8}{c}{\textbf{Joint Target States} \((l_{target}^1, l_{target}^2, l_{target}^3)\) \textbf{in Scenario \#2}} \\
\hline
(4, 5, 7)  & (4, 5, 11) & (10, 5, 7) & (10, 5, 11) & (8, 5, 7)  & (8, 5, 11) & (0, 1, 7)  & (0, 1, 11) \\
(4, 9, 7)  & (4, 9, 11) & (10, 9, 7) & (10, 9, 11) & (8, 9, 7)  & (8, 9, 11) & (1, 0, 7)  & (1, 0, 11) \\
\hline
\end{tabular}
    \vspace{-.2cm}
\caption{Target states in \textbf{Scenario \#2}. Joint local states of robots \(R1\), \(R2\), and \(R3\)}
\label{tab:Env_2_targets}
\end{table}

\begin{table}[H]
\centering
\begin{tabular}{ccccc @{\hskip 1cm} ccccc @{\hskip 1cm} ccccc}
\toprule
\cmidrule(r){1-5} \cmidrule(r){6-10} \cmidrule(r){11-15}
\textbf{Robot} & \textbf{Local} & \textbf{Action} & \textbf{Trans.} & \textbf{Next} &
\textbf{Robot} & \textbf{Local} & \textbf{Action} & \textbf{Trans.} & \textbf{Next} \\
& \textbf{State} &  & \textbf{Prob.} & \textbf{State} &
& \textbf{State} &  & \textbf{Prob.} & \textbf{State} &\\
\midrule
R1 & 2 & 0 & 0.5 & 0  & R2 & 2 & 3 & 1.0 & 0 \\
R1 & 2 & 0 & 0.5 & 1  & R2 & 2 & 0 & 1.0 & 1  \\
R1 & 3 & 0 & 1.0 & 2  & R2 & 3 & 3 & 1.0 & 2  \\
R1 & 3 & 2 & 1.0 & 6  & R2 & 4 & 0 & 1.0 & 3 \\
R1 & 5 & 3 & 1.0 & 4  & R2 & 4 & 1 & 0.2 & 5 \\
R1 & 6 & 3 & 1.0 & 5  & R2 & 4 & 1 & 0.8 & 9  \\
R1 & 6 & 1 & 1.0 & 7  & R3 & 8 & 3 & 0.8 & 7  \\
R1 & 6 & 2 & 0.7 & 9  & R3 & 8 & 3 & 0.2 & 11 \\
R1 & 6 & 2 & 0.2 & 10 &    &   &   &     &    \\
R1 & 6 & 2 & 0.1 & 11 &    &   &   &     &   \\
R1 & 7 & 1 & 1.0 & 8  &    &   &   &     &  \\
R1 & 9 & 1 & 1.0 & 10 &    &   &   &     &  \\
R1 & 11 & 3 & 1.0 & 10 &   &   &   &     &  \\
\bottomrule
\end{tabular}
\caption{\textbf{Scenario \# 2.} Transition probabilities for each robot as a function of state and action. The actions labeled as move North, move East, move South, move West, and Remain are indicated with 0, 1, 2, 3, and 4, respectively.}
\label{tab:Env_2_agent_transitions}
\end{table}

\begin{figure}[H]
    \centering
    \begin{subfigure}[t]{0.32\textwidth}
        \centering
        \includegraphics[width=\textwidth]{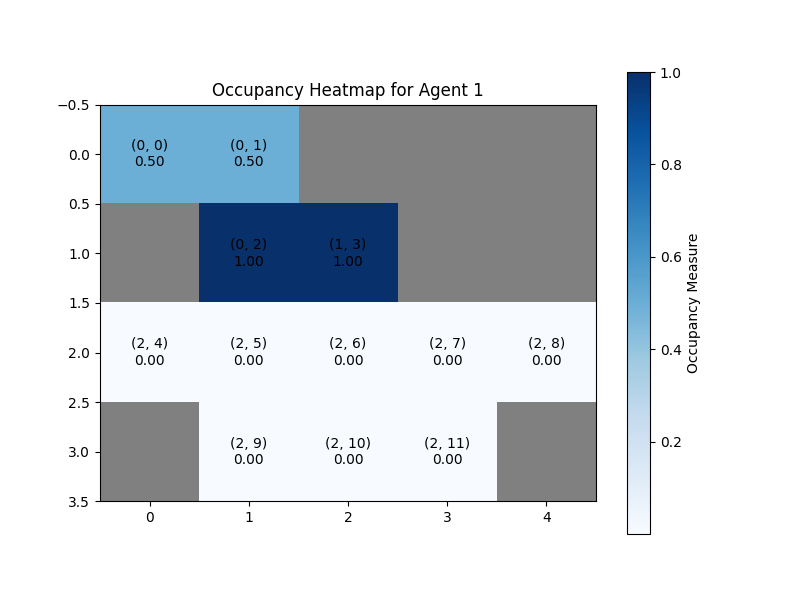}
        \caption{Occupancy measures for agent \(R1\) }
        \label{fig:subfig1_Env_2}
    \end{subfigure}
    \hspace{0.01\textwidth} 
    \begin{subfigure}[t]{0.32\textwidth}
        \centering
        \includegraphics[width=\textwidth]{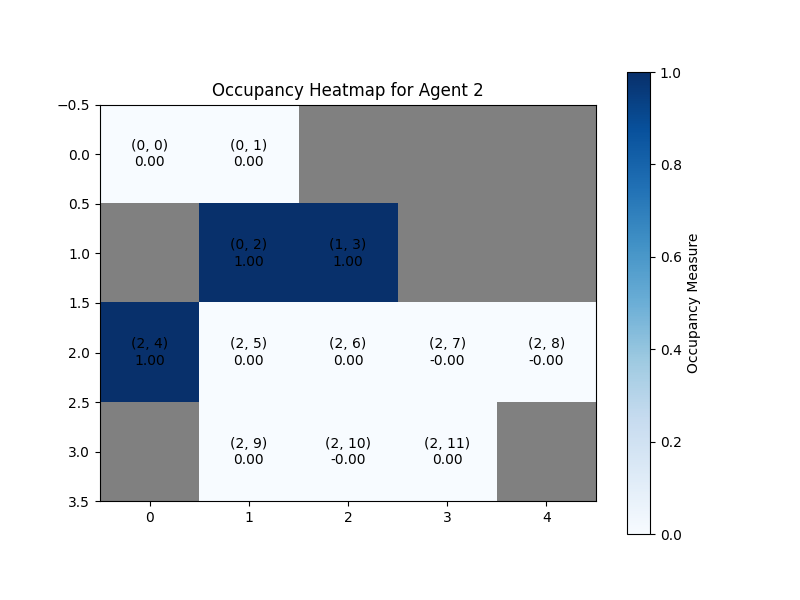}
        \caption{Occupancy measures for agent \(R2\)}
        \label{fig:subfig2_Env_2}
    \end{subfigure}
    \hspace{0.01\textwidth} 
    \begin{subfigure}[t]{0.32\textwidth}
        \centering
        \includegraphics[width=\textwidth]{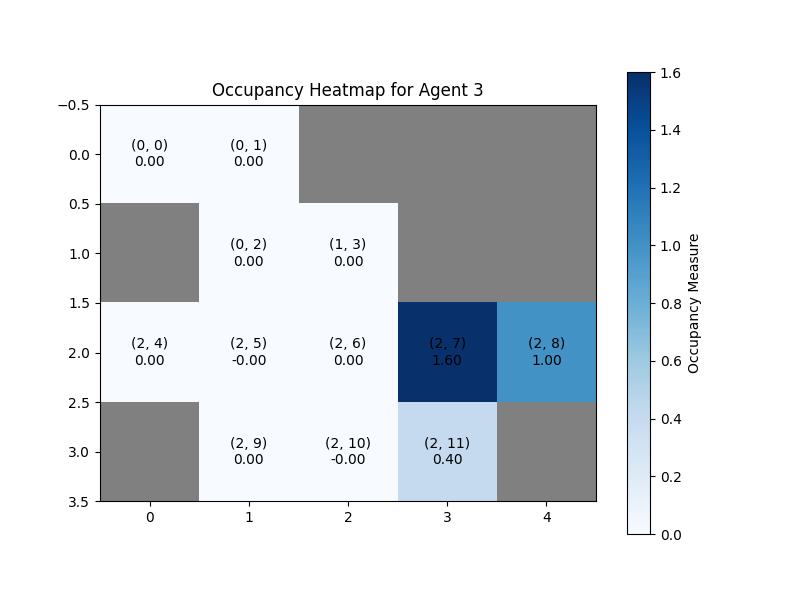}
        \caption{Occupancy measures for agent \(R3\)}
        \label{fig:subfig3_Env_2}
    \end{subfigure}
    \caption{\textbf{Scenario \#2.} Heat maps of the occupancy measures based on the joint action  policy computed  by solving (\ref{opt}).}
    \label{fig:Env_2_Heat_map_limited_comm}
\end{figure}

\begin{figure}[H]
    \centering
    \begin{subfigure}[t]{0.32\textwidth}
        \centering
        \includegraphics[width=\textwidth]{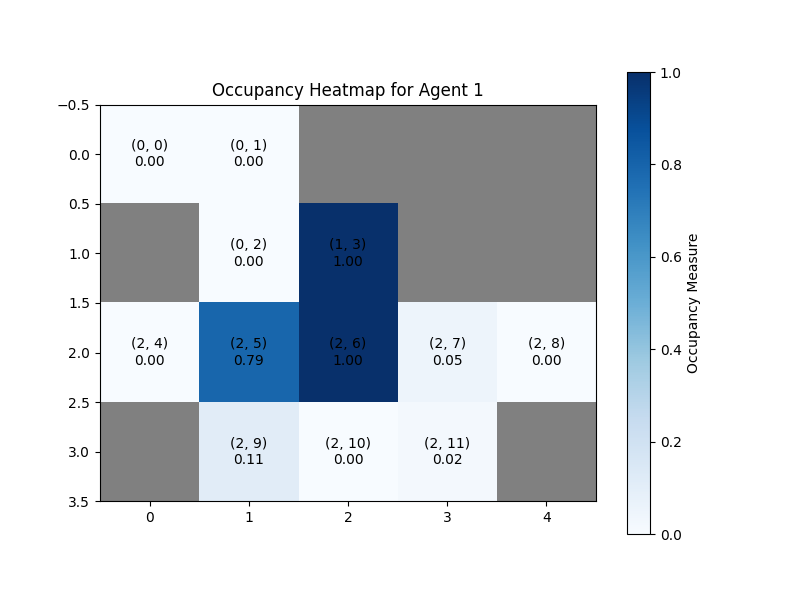}
        \caption{Occupancy measures for agent \(R1\) }
        \label{fig:subfig1_Env_2}
    \end{subfigure}
    \hspace{0.01\textwidth} 
    \begin{subfigure}[t]{0.32\textwidth}
        \centering
        \includegraphics[width=\textwidth]{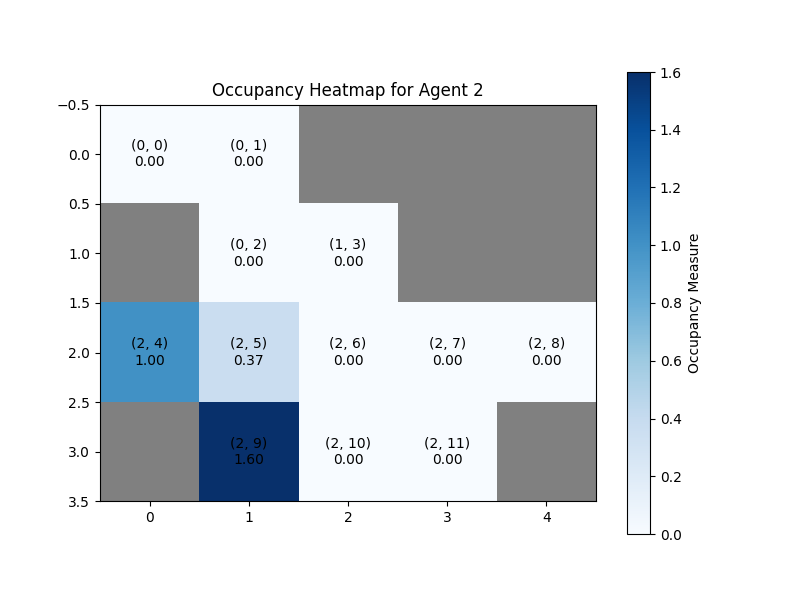}
        \caption{Occupancy measures for agent \(R2\)}
        \label{fig:subfig2_Env_2}
    \end{subfigure}
    \hspace{0.01\textwidth} 
    \begin{subfigure}[t]{0.32\textwidth}
        \centering
        \includegraphics[width=\textwidth]{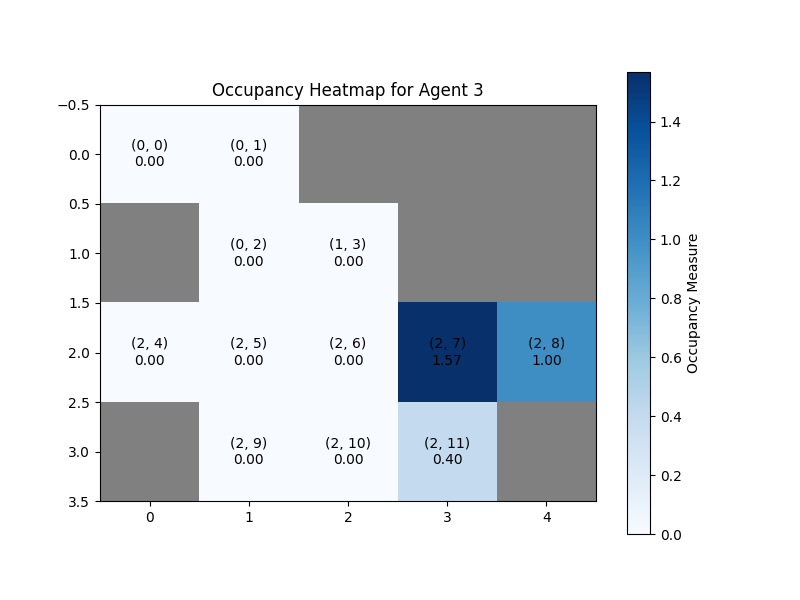}
        \caption{Occupancy measures for agent \(R3\)}
        \label{fig:subfig3_Env_2}
    \end{subfigure}
    \caption{\textbf{Scenario \#2.} Heat maps of the occupancy measures based on the joint action  policy computed  by minimizing total correlation.}
    \label{fig:Env_2_Heat_map_total_correlation}
\end{figure}

\clearpage

\subsubsection{Scenario \#3 with a Hallway}
Consider the environment in Figure~\ref{fig:subfig_map_Env_3}, where three robots are tasked with navigating to their respective goal locations, labeled \(T1\), \(T2\), and \(T3\). The environment consists of 9 cells, labeled from \(0\) to \(8\), as shown in Figure~\ref{fig:Env_3_local_states}, and is divided into three regions \(o=0\), \(o=1\), and \(o=2\), as illustrated in Figure~\ref{fig:subfig_region_Env_3}. A robot’s task is considered complete once it reaches one of its target locations. The objective is to reach these targets while avoiding collisions with the highest possible probability.
The transition probabilities for the movement of the robots can be found in Table~\ref{tab:Env_3_agent_transitions}. 
Note that coordination among the robots is critical at certain time steps to share local state and prevent collisions. 

We evaluate the policy computed using our approach, where at any given time step, only two out of the three robots are permitted to communicate and exchange precise location and state information. 
The heat maps of the occupancy measures for the robots under the synthesized joint policy are shown in
Figure~\ref{fig:Env_3_Heat_map_limited_comm}. 
The generated pair of action and communication policies is shown in Tables~\ref{tab:Env_3_action_policy}--\ref{tab:Env_3_communication_policy}.
Under this scenario, the suggested policy achieves a reach-avoid probability of \(1\). The communication policy dynamically adapts to changing public information, enabling robots to perform optimally without incurring any additional communication costs. This adaptability ensures that the communication cost remains zero while maintaining optimal task performance.

\begin{figure}
    \centering
    \includegraphics[scale=0.07]{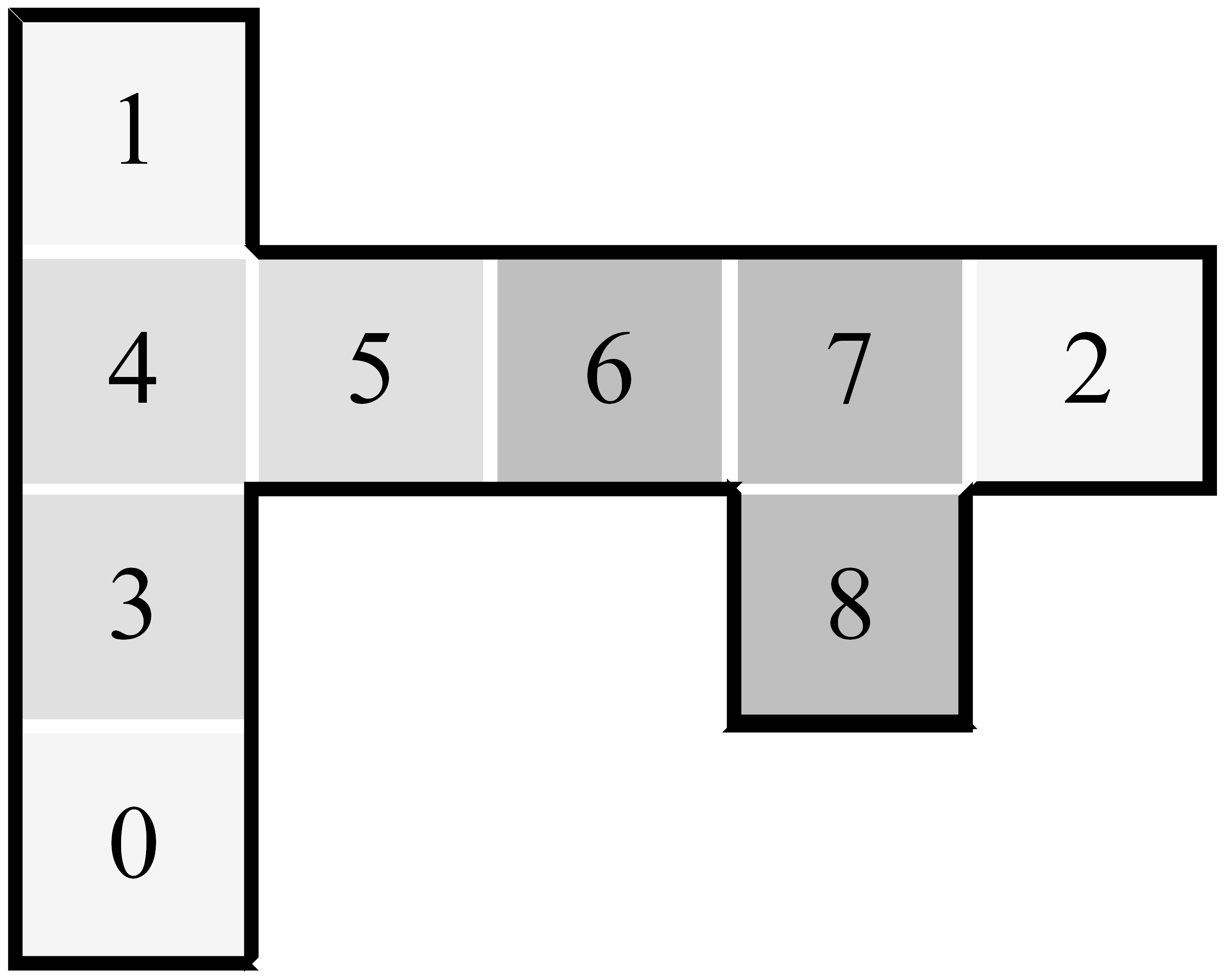}
    \caption{Local states labels used in \textbf{Scenario \#3}.}
    \label{fig:Env_3_local_states}
\end{figure}

\begin{figure}
    \centering
    \begin{subfigure}[t]{0.32\textwidth}
        \centering
        \includegraphics[width=\textwidth]{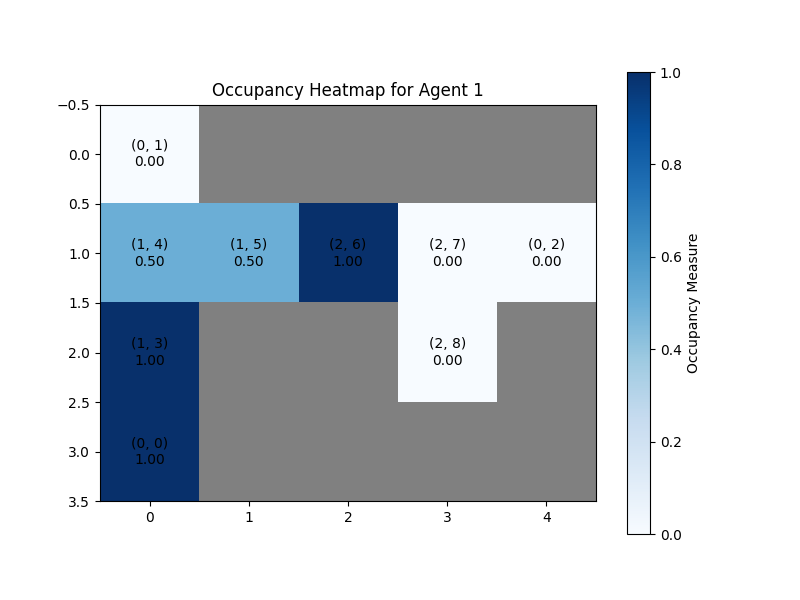}
        \caption{Occupancy measures for  agent \(R1\) }
        \label{fig:subfig1_Env_3}
    \end{subfigure}
    \hspace{0.01\textwidth} 
    \begin{subfigure}[t]{0.32\textwidth}
        \centering
        \includegraphics[width=\textwidth]{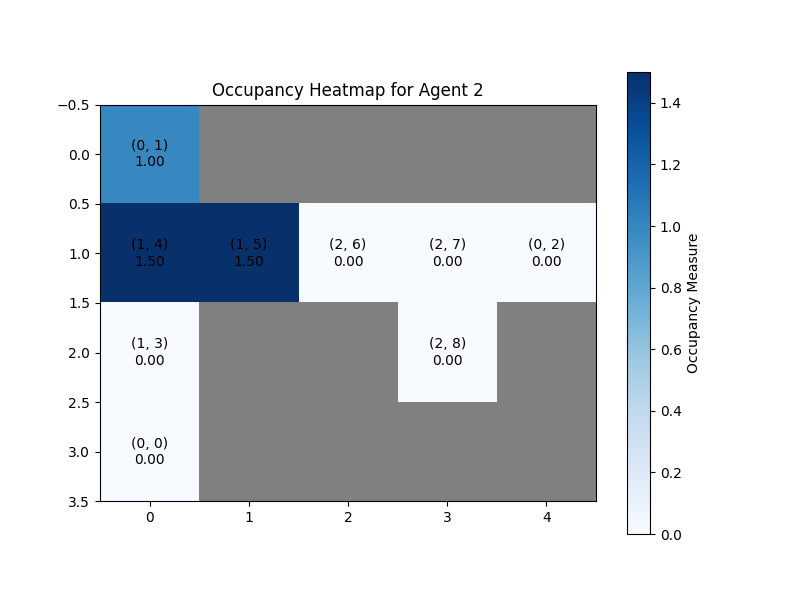}
        \caption{Occupancy measures for agent \(R2\)}
        \label{fig:subfig2_Env_3}
    \end{subfigure}
    \hspace{0.01\textwidth} 
    \begin{subfigure}[t]{0.32\textwidth}
        \centering
        \includegraphics[width=\textwidth]{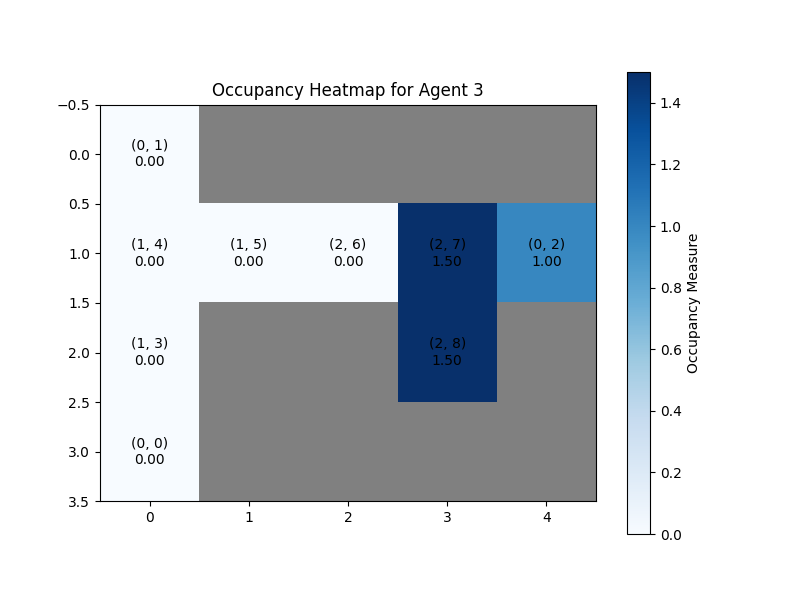}
        \caption{Occupancy measures for agent \(R3\)}
        \label{fig:subfig3_Env_3}
    \end{subfigure}
    \caption{\textbf{Scenario \#3.} Heat maps of the occupancy measures based on the joint action  policy computed  by solving (\ref{opt}).}
    \label{fig:Env_3_Heat_map_limited_comm}
\end{figure}

\begin{table}[h!]
\centering
\begin{tabular}{ccccc @{\hskip 1cm} ccccc @{\hskip 1cm} ccccc}
\toprule
\cmidrule(r){1-5} \cmidrule(r){6-10} \cmidrule(r){11-15}
\textbf{Robot} & \textbf{Local} & \textbf{Action} & \textbf{Trans.} & \textbf{Next} &
\textbf{Robot} & \textbf{Local} & \textbf{Action} & \textbf{Trans.} & \textbf{Next} \\
& \textbf{State} &  & \textbf{Prob.} & \textbf{State} &
& \textbf{State} &  & \textbf{Prob.} & \textbf{State} &\\
\midrule
R1 & 0 & 0  & 1.0 & 3 & R2 & 1 & 2  & 0.5 & 4 \\ 
R1 & 3 & 0  & 1.0 & 4 & R2 & 1 & 2  & 0.5 & 5 \\
R1 & 3 & 1  & 1.0 & 5 & R2 & 4 & 4  & 1.0 & 4 \\
R1 & 4 & 0  & 1.0 & 6 & R2 & 5 & 4  & 1.0 & 5 \\
R1 & 5 & 1  & 1.0 & 6 & R3 & 2 & 3  & 0.5 & 7 \\
R1 & 6 & 1  & 1.0 & 7 & R3 & 2 & 3  & 0.5 & 8 \\
R1 & 6 & 2  & 1.0 & 8 & R3 & 7 & 4 & 1.0 & 7  \\
R1 & 7 & 4  & 1.0 & 7 & R3 & 8 & 4 & 1.0 & 8 \\
R1 & 8 & 4  & 1.0 & 8 &    &   &   &     &    \\
\bottomrule
\end{tabular}
\caption{\textbf{Scenario \# 3.} Transition Probabilities for each robot as a function of state and action. The actions labeled as move North, move East, move South, move West, and Remain are indicated with 0, 1, 2, 3, and 4, respectively.}
\label{tab:Env_3_agent_transitions}
\end{table}
\clearpage

\begin{table}[h!]
    \centering
    \begin{tabular}{ccc} 
        \toprule
        \textbf{Joint Local State} \\ \((l^1, l^2, l^3)\) & \textbf{Action} & \textbf{Probability} \\
        \midrule
        (0, 1, 2)   & (0, 2, 3) & 1.00 \\
        (3, 4, 7)   & (1, 4, 4) & 1.00 \\
        (3, 4, 8)   & (1, 4, 4) & 1.00 \\
        (5, 4, 7)   & (1, 4, 4) & 1.00 \\
        (5, 4, 8)   & (1, 4, 4) & 1.00 \\
        (6, 4, 7)   & (2, 4, 4) & 1.00 \\
        (6, 4, 8)   & (1, 4, 4) & 1.00 \\
        (3, 5, 7)   & (0, 4, 4) & 1.00 \\
        (3, 5, 8)   & (0, 4, 4) & 1.00 \\
        (4, 5, 7)   & (0, 4, 4) & 1.00 \\
        (4, 5, 8)   & (0, 4, 4) & 1.00 \\
        (6, 5, 7)   & (2, 4, 4) & 1.00 \\
        (6, 5, 8)   & (1, 4, 4) & 1.00 \\
        \bottomrule
    \end{tabular}
    \caption{\textbf{Scenario \# 3.} Action policy: The actions are labeled as move North, move East, move South, move West, and Remain with 0, 1, 2, 3, and 4, respectively.}
    \label{tab:Env_3_action_policy}
\end{table}

\begin{table}[h!]
    \centering
    \begin{tabular}{ccc} 
        \toprule
        \textbf{Joint Public} & \textbf{Communication} & \multirow{2}{*}{\textbf{Probability}}\\
        \textbf{Information} \((o^1, o^2, o^3)\) & \textbf{Action} & \\
        \midrule
        (1, 1, 2) & Robot 1 and 2 & 1 \\
        (2, 1, 2) & Robot 1 and 3 & 1 \\
        (0, 0, 0) & Robot 1 and 2 & 1 \\
        \bottomrule
    \end{tabular}
    \caption{\textbf{Scenario \# 3.} Communication policy as a function of joint public information with the probabilities.}
    \label{tab:Env_3_communication_policy}
\end{table}
\clearpage

\subsubsection{Scenario \#4 with High Uncertainty}
Consider a 3 x 3 grid environment as in Figure~\ref{fig:Environment_4} with three robots $R1, R2,R3$ and target locations $T1, T2, T3$, respectively.
The robots must navigate to their respective target locations while avoiding collisions. The robot can communicate which row they are in.
Each robot has five possible actions at any given time: moving North, East, South, West, or remaining in its current position. If the chosen movement is valid (i.e., stays within the grid boundaries), the robot transitions to the intended neighboring cell with a probability of \(0.9\), while the remaining \(0.1\) slip probability is redistributed across the current cell and all other valid neighboring cells. If the intended movement is invalid (i.e., leads outside the grid boundaries), the full transition probability (\(1.0\)) is redistributed among the current cell and all valid neighboring cells.
The team's objective is to reach the target locations while avoiding collisions with the highest probability. In this scenario, all three robots are allowed to share publicly observable parts of each state, while only two out of the three robots can fully communicate at each step, sharing the local parts of their current states. 

Figures~\ref{fig:Env_4_Heat_map_full_comm}--\ref{fig:Env_4_Heat_map_limited_comm} present the heat maps of the occupancy measures for the robots  for the joint action policy synthesized 
without and with minimizing communication respectively.
Under full communication, the team can complete its task with a maximum probability of \(0.958\). Under restricted communication, while no optimal action and communication policy with zero communication cost exists for achieving the maximum reach-avoid probability, our method can compute a pair of action and communication policies with zero communication cost for a lower threshold of the reach avoid probability, which is \(0.92\).

\begin{figure}
    \centering
    \begin{subfigure}[t]{0.32\textwidth}
        \centering
        \includegraphics[width=\textwidth]{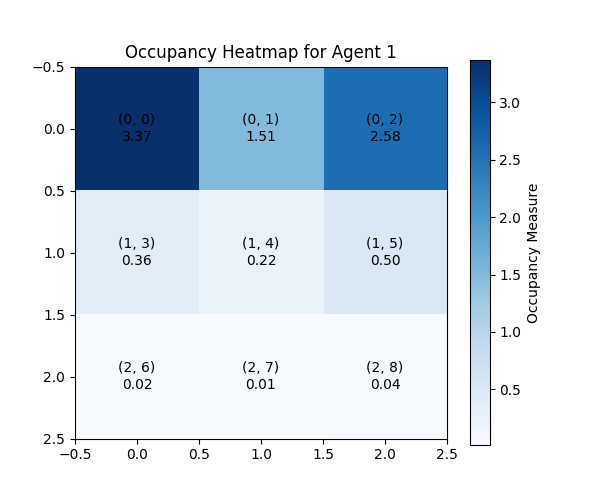}
        \caption{Occupancy measures for agent \(R1\) }
        \label{fig:subfig1_Env_4}
    \end{subfigure}
    \hspace{0.01\textwidth} 
    \begin{subfigure}[t]{0.32\textwidth}
        \centering
        \includegraphics[width=\textwidth]{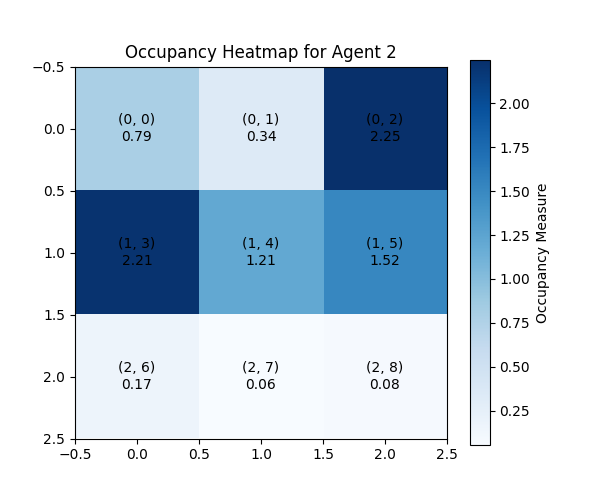}
        \caption{Occupancy measures for agent \(R2\)}
        \label{fig:subfig2_Env_4}
    \end{subfigure}
    \hspace{0.01\textwidth} 
    \begin{subfigure}[t]{0.32\textwidth}
        \centering
        \includegraphics[width=\textwidth]{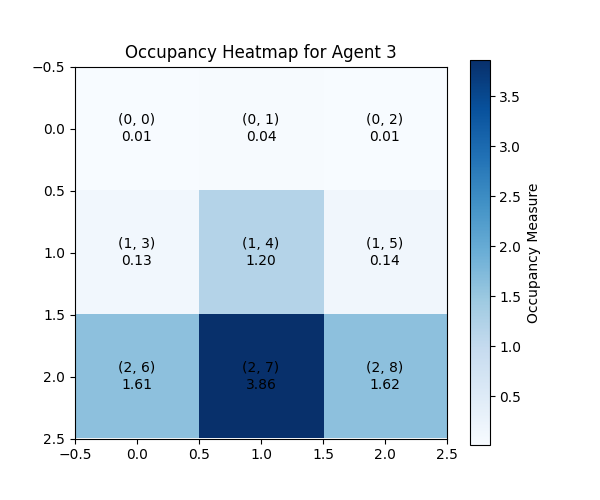}
        \caption{Occupancy measures  for agent \(R3\)}
        \label{fig:subfig3_Env_4}
    \end{subfigure}
    \caption{\textbf{Scenario \#4.} Heat maps of the occupancy measures for the policy computed without minimizing communication.}
    \label{fig:Env_4_Heat_map_full_comm}
\end{figure}

\begin{figure}
    \centering
    \begin{subfigure}[t]{0.32\textwidth}
        \centering
        \includegraphics[width=\textwidth]{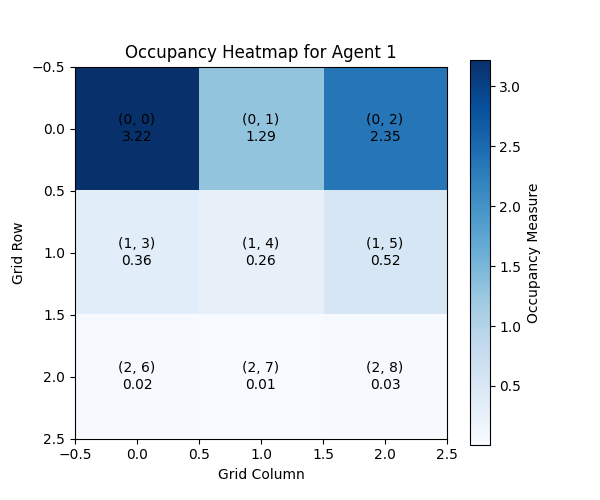}
        \caption{Occupancy measures for agent \(R1\) }
        \label{fig:subfig4_Env_4}
    \end{subfigure}
    \hspace{0.01\textwidth} 
    \begin{subfigure}[t]{0.32\textwidth}
        \centering
        \includegraphics[width=\textwidth]{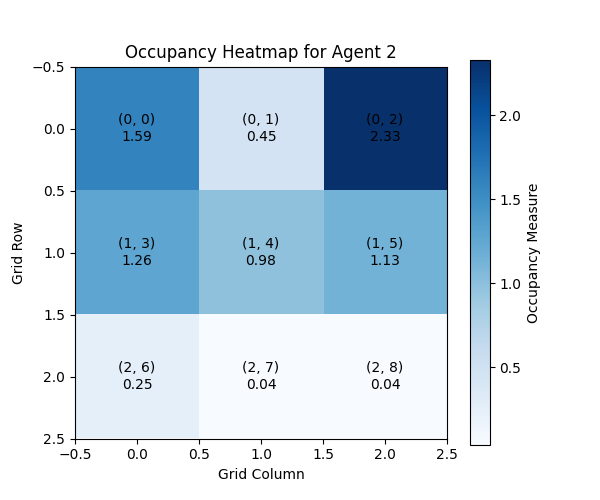}
        \caption{Occupancy measures for agent \(R2\)}
        \label{fig:subfig5_Env_1}
    \end{subfigure}
    \hspace{0.01\textwidth} 
    \begin{subfigure}[t]{0.32\textwidth}
        \centering
        \includegraphics[width=\textwidth]{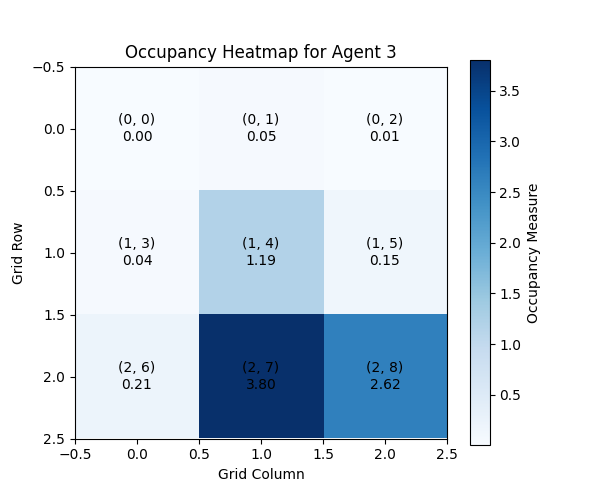}
        \caption{Occupancy measures for agent \(R3\)}
        \label{fig:subfig6_Env_4}
    \end{subfigure}
    \caption{\textbf{Scenario \#4.} Heat maps of the occupancy measures based on the joint action  policy computed  by solving (\ref{opt}).}
    \label{fig:Env_4_Heat_map_limited_comm}
\end{figure}

\end{document}